\documentclass[twoside,11pt,final]{entics}

\usepackage{enticsmacro}
\usepackage{graphicx}
\usepackage[all]{xy}
\usepackage{amsmath}
\usepackage[all]{xy}
\usepackage{hyperref}
\usepackage{cleveref}
\usepackage{proof}
\usepackage{xspace}
\usepackage[normalem]{ulem}

\usepackage{algorithm}
\usepackage{algpseudocode}
\usepackage{bm}
\usepackage{cancel}

\usepackage{tikz}
\usetikzlibrary{positioning,arrows.meta}

\usepackage{thm-restate}
\newtheorem{obs}[thm]{Observation}
\newtheorem{notation}[thm]{Notation}

\Crefname{defn}{Definition}{Definitions}
\Crefname{thm}{Theorem}{Theorems}
\Crefname{rem}{Remark}{Remarks}
\Crefname{lem}{Lemma}{Lemmas}
\Crefname{prop}{Proposition}{Propositions}
\Crefname{obs}{Observation}{Observations}
\Crefname{cor}{Corollary}{Corollaries}

\crefname{defn}{Definition}{Definitions}
\crefname{thm}{Theorem}{Theorems}
\crefname{rem}{Remark}{Remarks}
\crefname{lem}{Lemma}{Lemmas}
\crefname{prop}{Proposition}{Propositions}
\crefname{obs}{Observation}{Observations}
\crefname{cor}{Corollary}{Corollaries}

\newcommand{\cris}[1]{{\color{blue}{#1}}}

\renewcommand{\theenumi}{\arabic{enumi}}
\renewcommand{\theenumii}{\arabic{enumii}}
\renewcommand{\theenumiii}{\arabic{enumiii}}

\makeatletter
\renewcommand\p@enumii{\theenumi.}
\renewcommand\p@enumiii{\theenumi.\theenumii.}
\renewcommand\p@enumiv{\theenumi.\theenumii.\theenumiii.}
\makeatother

\newcommand{\arrowvec}[1]{\mathaccent"017E{#1}}

\newcommand{\defnn}[1]{{\bf #1}}
\newcommand{\ie}{{\em i.e.}\xspace}
\newcommand{\eg}{{\em e.g.}\xspace}
\newcommand{\ih}{IH\xspace}

\newcommand{\ST}{\ |\ }
\renewcommand{\emptyset}{\varnothing}
\newcommand{\set}[1]{\{#1\}}

\newcommand{\eqdef}{:=}

\newcommand{\itemNumber}[1]{\textsf{\textup{\textbf{#1}}}}
\newcommand{\lemmaPart}[1]{\textsf{\textup{\textbf{#1}}}}
\newcommand{\maxof}[1]{\max\set{#1}}

\newcommand{\lam}[2]{\lambda#1.\,#2}
\newcommand{\app}[2]{#1\,#2}
\newcommand{\sub}[2]{[#1:=#2]}
\newcommand{\fv}[1]{\mathsf{fv}(#1)}
\newcommand{\fvtyp}[2]{\mathsf{fv}_{#1}(#2)}
\newcommand{\wrap}[1]{\bm{\langle}#1\bm{\rangle}}
\newcommand{\bin}[2]{#1\wrap{#2}}
\newcommand{\wlist}[1][\tms]{\ctxhole\wrap{#1_1}\hdots\wrap{#1_n}}

\newcommand{\Lamesymbol}{\mathtt{e}}
\newcommand{\LamTsymbol}{\mathtt{i}}
\newcommand{\LamTmsymbol}{\mathtt{im}}
\newcommand{\CalcLam}{\ensuremath{\Lambda}}
\newcommand{\CalcLame}{\ensuremath{\Lambda^{\Lamesymbol}_\cap}}
\newcommand{\CalcLamT}{\ensuremath{\Lambda^{\LamTsymbol}_\cap}}
\newcommand{\CalcLamTm}{\ensuremath{\Lambda^{\LamTmsymbol}_\cap}}

\newcommand{\tobeta}{\to_\beta}
\newcommand{\symT}{\LamTsymbol}
\newcommand{\symTm}{\LamTmsymbol}
\newcommand{\toT}{\to_\symT}
\newcommand{\toTs}{\to_\symT^*}
\newcommand{\pT}{\Rightarrow_\symT}
\newcommand{\pTm}{\Rightarrow_\symTm}
\newcommand{\toTm}{\to_\symTm}
\newcommand{\tof}{\rhd}
\newcommand{\tofvert}{\triangledown}

\newcommand{\rulee}[1]{\Lamesymbol\text{-}\mathtt{#1}}
\newcommand{\ruleeAx}{\rulee{\ensuremath{\mathtt{var}}}}
\newcommand{\ruleeArri}{\rulee{\mathtt{I{\arrow}}}}
\newcommand{\ruleeArre}{\rulee{\mathtt{E{\arrow}}}}
\newcommand{\ruleeMulti}{\rulee{\mathtt{many}}}
\newcommand{\ruleT}[1]{\LamTsymbol\text{-}\mathtt{#1}}
\newcommand{\ruleTAx}{\ruleT{\mathtt{var}}}
\newcommand{\ruleTArri}{\ruleT{\mathtt{I{\arrow}}}}
\newcommand{\ruleTArre}{\ruleT{\mathtt{E{\arrow}}}}
\newcommand{\ruleTMulti}{\ruleT{\mathtt{many}}}
\newcommand{\ruleTm}[1]{\LamTmsymbol\text{-}\mathtt{#1}}
\newcommand{\ruleTmAx}{\ruleTm{\mathtt{var}}}
\newcommand{\ruleTmArri}{\ruleTm{\mathtt{I{\arrow}}}}
\newcommand{\ruleTmArre}{\ruleTm{\mathtt{E{\arrow}}}}
\newcommand{\ruleTmWrap}{\ruleTm{\mathtt{wrap}}}
\newcommand{\ruleTmMulti}{\ruleTm{\mathtt{many}}}

\newcommand{\dom}[1]{\mathsf{dom}(#1)}
\newcommand{\tctx}{\Gamma}
\newcommand{\tctxtwo}{\Delta}

\newcommand{\MinCtx}[1]{\mathsf{\gamma}(#1)}

\newcommand{\ctxhole}{\Box}
\newcommand{\sctx}{\mathtt{L}}

\newcommand{\itoN}{{i \in 1..n}}

\newcommand{\iI}{{i \in I}}
\newcommand{\jJ}{{j \in J}}

\newcommand{\var}{x}
\newcommand{\vartwo}{y}
\newcommand{\varthree}{z}

\newcommand{\utm}{M}
\newcommand{\utmtwo}{N}
\newcommand{\utmthree}{P}
\newcommand{\utmfour}{Q}

\newcommand{\tm}{t}
\newcommand{\tmtwo}{s}
\newcommand{\tmthree}{u}
\newcommand{\tmfour}{r}

\newcommand{\tms}{{\arrowvec{\tm}}}
\newcommand{\tmstwo}{{\arrowvec{\tmtwo}}}
\newcommand{\tmsthree}{{\arrowvec{\tmthree}}}
\newcommand{\tmsfour}{{\arrowvec{\tmfour}}}

\newcommand{\basetyp}{a}
\newcommand{\basetyptwo}{b}

\newcommand{\arrow}{\rightarrow}
\newcommand{\typ}{A}
\newcommand{\typtwo}{B}
\newcommand{\typthree}{C}
\newcommand{\typfour}{D}

\newcommand{\typs}{{\arrowvec{\typ}}}
\newcommand{\typstwo}{{\arrowvec{\typtwo}}}
\newcommand{\typsthree}{{\arrowvec{\typthree}}}

\newcommand{\judge}[3]{#1\vdash_{\Lamesymbol}#2:#3}
\newcommand{\judges}[3]{#1\Vdash_{\Lamesymbol}#2:#3}
\newcommand{\judgT}[3]{#1\vdash_{\LamTsymbol}#2:#3}
\newcommand{\judgTs}[3]{#1\Vdash_{\LamTsymbol}#2:#3}
\newcommand{\judgTm}[3]{#1\vdash_{\LamTmsymbol}#2:#3}
\newcommand{\judgTms}[3]{#1\Vdash_{\LamTmsymbol}#2:#3}
\newcommand{\judgRefine}[2]{#1 \sqsubset #2}

\newcommand{\derivs}[2]{#1\mathrel{\rhd}#2}
\newcommand{\deriv}{\Pi}
\newcommand{\judgment}{\mathcal{J}}
\newcommand{\typdec}[1]{#1^{\mathtt{i}}}
\newcommand{\typerase}[1]{#1^{\mathtt{e}}}

\newcommand{\typeof}[1]{\mathsf{type}(#1)}
\newcommand{\height}[1]{\mathsf{h}(#1)}
\newcommand{\weight}[1]{\mathsf{w}(#1)}
\newcommand{\maxdeg}[1]{\mathsf{\delta}_{\mathsf{max}}(#1)}

\newcommand{\simp}[2]{\mathtt{S}_{#1}(#2)}
\newcommand{\simpd}[1]{\simp{d}{#1}}
\newcommand{\simpfull}[1]{\simp{*}{#1}}
\newcommand{\comp}[1]{\mathtt{C}{(#1)}}

\newcommand{\meassym}{\mathcal{W}}
\newcommand{\meas}[1]{\meassym(#1)}

\newcommand{\diagramCommuteTmsFpTmsFp}[4]{
  \xymatrix@C=.2em@R=.2em{
    #1         & \toTm^* & #2         \\
    \tofvert^+ &         & \tofvert^+ \\
    #3         & \toTm^* & #4         \\
  }
}
\newcommand{\diagramCommuteTmFTmFp}[4]{
  \xymatrix@C=.2em@R=.2em{
    #1        & \toTm   & #2         \\
    \tofvert  &         & \tofvert^+ \\
    #3        & \toTm^= & #4         \\
  }
}
\newcommand{\diagramCommuteTmFTmFOne}[4]{
  \xymatrix@C=.2em@R=.2em{
    #1        & \toTm   & #2         \\
    \tofvert  &         & \tofvert \\
    #3        & \toTm^= & #4         \\
  }
}
\newcommand{\diagramCommuteTmFTmOneFp}[4]{
  \xymatrix@C=.2em@R=.2em{
    #1        & \toTm   & #2         \\
    \tofvert  &         & \tofvert^+ \\
    #3        & \toTm & #4         \\
  }
}
\newcommand{\diagramCommuteTmFTmOneFOne}[4]{
  \xymatrix@C=.2em@R=.2em{
    #1        & \toTm   & #2         \\
    \tofvert  &         & \tofvert \\
    #3        & \toTm & #4         \\
  }
}
\newcommand{\diagramCommuteTmFEqFOne}[4]{
  \xymatrix@C=.2em@R=.2em{
    #1        & \toTm   & #2         \\
    \tofvert  &         & \tofvert \\
    #3        & = & #4         \\
  }
}

\newcommand{\qedhere}{}

\def\conf{MFPS 2025} 	
\volume{5}			
\def\volu{5}			
\def\papno{3}			
\def\mydoi{16693}
\def\lastname{Barenbaum, Ronchi Della Rocca, Sottile}
\def\runtitle{Strong normalization through idempotent intersection types}
\def\copynames{P. Barenbaum, S. Ronchi Della Rocca,}
\begin{document}
\begin{frontmatter}
  \title{Strong Normalization Through Idempotent Intersection Types: A New
    Syntactical Approach\thanksref{P}\thanksref{ALL}}
    \thanks[P]{Funded by the European Union through the MSCA SE project QCOMICAL (Grant Agreement ID: 101182520).}
    \thanks[ALL]{We would like to thank the anonymous reviewers for their careful reading of the manuscript and their insightful comments, which helped improve the quality of this paper.}

  \author{Pablo Barenbaum\thanksref{c}\thanksref{b}\thanksref{emailpb}} 
  \author{Simona Ronchi Della Rocca\thanksref{a}\thanksref{emailsr}} 
  \author{Cristian Sottile\thanksref{b}\thanksref{c}\thanksref{emailcs}} 

  \address[a]{Dipartimento di Informatica, Università di Torino, Italy}
  \address[b]{ICC, CONICET—Universidad de Buenos Aires, Argentina}
  \address[c]{Universidad Nacional de Quilmes, Argentina}

  \thanks[emailpb]{Email: \href{mailto:pbarenbaum@dc.uba.ar}{\texttt{pbarenbaum@dc.uba.ar}}}
  \thanks[emailsr]{Email: \href{mailto:ronchi@di.unito.it}{\texttt{ronchi@di.unito.it}}}
  \thanks[emailcs]{Email: \href{mailto:csottile@dc.uba.ar}{\texttt{csottile@dc.uba.ar}}}

\begin{abstract}
  
It is well-known that intersection type assignment systems can be used to characterize strong normalization (SN). Typical proofs that typable lambda-terms are SN in these systems rely on \emph{semantical} techniques. In this work, we study $\CalcLame$, a variant of Coppo and Dezani's (Curry-style) intersection type system and we propose a syntactical proof of strong normalization for it. We first design $\CalcLamT$, a Church-style version of $\CalcLame$, in which terms closely correspond to typing derivations.
We then prove that typability in $\CalcLamT$ implies SN through a measure that, given a term, produces a natural number that decreases along with reduction. This measure provides a \emph{syntactical} proof of SN. Finally, the result is extended to $\CalcLame$, since the two systems simulate each other.

\end{abstract}


\end{frontmatter}

\section{Introduction}
\label{intro}
Intersection types were introduced by Coppo and Dezani
\cite{DBLP:journals/ndjfl/CoppoD80}, with the explicit aim of increasing the
typability power of simple types. Beyond the arrow connective of simple types
($\arrow$), these systems incorporate an intersection connective ($\wedge$),
which is commutative, associative and idempotent (\ie $\typ \wedge \typ \equiv
\typ$), which means that intersection is essentially a notation
for \emph{sets} of types.

Intersection type assignment systems come in different versions, and are able to characterize semantical properties of terms.
For example, Coppo and Dezani's original system \emph{characterizes} strong normalization,
in the sense that a term is typable if and only if it is strongly normalizing.
Recall that a $\lambda$-term is \emph{strongly normalizing} if every reduction sequence starting from it eventually terminates.

There are many other intersection type systems characterizing other properties besides strong normalization, such as head normalization \cite{DBLP:conf/icalp/CoppoDS79,DBLP:books/daglib/0071545}, where we recall that a $\lambda$-term is \emph{head normalizing} if head reduction (reducing at every step the redex in head position) eventually terminates.
Furthermore, intersection types are a very powerful tool to reason about the interpretation of terms in various models of $\lambda$-calculus \cite{DBLP:journals/jsyml/BarendregtCD83,DBLP:series/txtcs/RoccaP04,barendregt2013lambda}.

The existent proofs of the characterization properties are not straightforward:
for both strong and head normalization, the direction ``typable implies
normalizing'' is generally carried out using semantical tools like computability
or reducibility candidates \cite{Pottinger80,DBLP:books/daglib/0071545}. These
semantical proofs are in some sense standard, but it is not easy to extract from
them an intuition about what is going on. Having an explicit measure that
decreases with $\beta$-reduction would be preferable. In the case of strong
normalization, there are three completely syntactical proofs, based on a
decreasing
measure~\cite{DBLP:conf/lics/KfouryW95,DBLP:conf/tlca/Boudol03,DBLP:journals/mscs/BucciarelliPS03};
we will briefly discuss them and their relation with the present work at the end
of this section.

A non-idempotent version of intersection types has been defined \cite{Gardner94,Kfoury00,deC09},
where the intersection connective is \emph{not} idempotent (\ie $\typ \wedge
\typ \not\equiv \typ$), and thus becomes a notation for \emph{multisets} of
types. The non-idempotent systems preserve some of the properties of the
idempotent ones: they can characterize strong and head normalization
\cite{BL11}. From a semantical point of view, they can describe relational
models of $\lambda$-calculus \cite{DBLP:journals/mscs/PaoliniPR17}; even if this
class of models is quite poor, with respect to the $\lambda$-theories that they
induce, they are useful to explore quantitative properties of typable terms.
Unlike what happens in the idempotent case, the proof that these systems with
non-idempotent intersection characterize strong and head normalization can be
carried out by an easy induction on the size of the typing derivation, which
decreases at every reduction step. This relies on the quantitative nature of
non-idempotency, and in particular on the fact that, in a derivation, the
cardinality of the multiset of types of a variable $\var$ is an upper bound of the number of
occurrences of $\var$ in the typed term. Given that these proofs
rely crucially on the fact that intersection is non-idempotent, there is not
much hope of extending them to the idempotent case. But we think that it
would be interesting to explore if there is a decreasing measure for idempotent
intersection types which is simpler than those in~\cite{DBLP:conf/lics/KfouryW95,DBLP:conf/tlca/Boudol03}.
Here we present a positive answer to this
problem, namely a decreasing measure consisting of just a natural number.
\medskip \noindent

\subsection{Contributions} Recently, a proof of strong normalization for simply typed
$\lambda$-calculus has been presented~\cite{BarenbaumS23}, based on a decreasing measure which can be defined
constructively, by relying on a non-erasing notion of
reduction inspired by \emph{memory
calculi}~\cite{NederpeltPhdThesis,KlopPhdThesis,Groote93,NeergaardS02,BucciarelliKV16}.
Namely, each $\beta$-reduction step creates a ``wrapper'' containing a copy of
the argument, remembering all the erased subterms. An operation
called \emph{full simplification} is defined
to repeatedly perform the complete development of all the redexes of
maximum degree in a term\footnote{ In this context, the \emph{degree} of a redex
$(\lam{\var}{\tm})\,\tmtwo$ is the height of the type of the abstraction
$\lam{\var}{\tm}$.}. A measure decreasing at every reduction step is then
given simply by a natural number corresponding to the number of wrappers that
remain after full simplification.
Here we \textbf{adapt this technique to idempotent intersection types}.

The first step in order to do this job is to design an
intrinsically typed version of intersection types, which are usually presented
as extrinsically typed systems instead. Let us recall that, in an
\emph{extrinsically typed} (Curry-style) system,
terms are untyped and there is a set of typing
rules assigning types to them, whereas in an
\emph{intrinsically typed} (Church-style) system,
terms are decorated with types, and typing rules are just a check
of syntactical correctness. There are some attempts in the literature in this
direction, \eg in~\cite{BVB08,DBLP:journals/iandc/LiquoriR07,DBLP:conf/birthday/BarbaneraDdV24}, but they are not suitable for the application
of the discussed techniques. In fact the big difference between simple and
intersection types, on which the power of the latter relies, is that, if $\utm$
is a term of $\lambda$-calculus typable by intersection types and $\utmtwo$ is a
subterm of $\utm$ occurring once in it, a typing derivation of a judgement
$\judge{\tctx}{\utm}{\typ}$ can contain $n \geq 1$ subderivations with object
$\utmtwo$. Since $\utmtwo$ can contain redexes,
performing a reduction inside $\utmtwo$ corresponds to performing $n$ reductions
\emph{in parallel} on the derivation.
The natural way to achieve this is to equip the Church-style intersection type
system with a parallel notion of reduction, and this is indeed what all the
existing proposals do.
At the same time, in order to apply the cited method, we would need to
reduce these subderivations independently from each other.

Our starting point is an \textbf{e}xtrinsically typed subset of $\Lambda$,
called $\CalcLame$, defined through an intersection type assignment system.
This system is a variant of well-known systems in the literature that characterize strong
normalization.
In order to build an intrinsically typed version of  $\CalcLame$, we exploit
the idea that idempotent intersection represents a set
of types, and we explicitly use sets both in the grammar of
types and terms of the typed language, so for example if $\tm$ and $\tmtwo$ are
terms of type $\typ$ and $\typtwo$ respectively, then $\set{\tm,\tmtwo}$ is a
set of terms which can be assigned the set of types $\set{\typ,\typtwo}$.
Following this idea, we formulate
an \textbf{i}ntrinsically typed
version of $\CalcLame$ called $\CalcLamT$, where
sets of types are decorations for sets of terms.
While these two systems can be dealt with independently from each other,
they are shown to \emph{simulate} each other.
The new typed language $\CalcLamT$
enjoys all the good properties we expect, such as subject reduction
and confluence.
More importantly, it is suitable for the application of a method extending that of
\cite{BarenbaumS23}. As a matter of fact, using the notion of complete
development by degree, and relying on an auxiliary \emph{memory calculus}
$\CalcLamTm$, we define a measure which is a natural
number decreasing at every reduction step in $\CalcLamT$. Then, we prove
that, for every term $\utm$ in $\CalcLame$, there is a term $\tm$ in $\CalcLamT$
such that an infinite reduction sequence in $\utm$ entails an infinite reduction
sequence in $\tm$, thus obtaining a new proof of strong normalization for
$\CalcLame$. Finally, we complete the picture by proving that $\CalcLamT$ gives type
to all the strongly normalizing terms, hence showing that it has the same typability power of the original 
paper \cite{DBLP:journals/ndjfl/CoppoD80}.

\subsection{Comparison with related works}
As we recalled before, there are three other syntactical proofs of
strong normalization for idempotent intersection types,
namely~\cite{DBLP:conf/lics/KfouryW95,DBLP:conf/tlca/Boudol03,DBLP:journals/mscs/BucciarelliPS03}.
Both~\cite{DBLP:conf/lics/KfouryW95,DBLP:conf/tlca/Boudol03} define an
intermediate language, typed by intersection types, equipped with a suitable
non-erasing reduction rule. So, as known for $\lambda I$ (\ie non–erasing) based
systems~\cite[\S 11.3]{barendregt1984}, weak normalization there entails strong
normalization. Using the property that the normalization of both the reductions
implies $\beta$-normalization, they prove weak normalization for terms typable in
the intersection type assignment system, and obtain strong normalization as a
corollary. In both papers, the normalization proofs are carried out by defining
a measure that decreases when a particular reduction strategy is chosen. The
approach of reducing weak to strong normalization can be traced back to at least
Nederpelt's and Klop's theses~\cite{NederpeltPhdThesis,KlopPhdThesis}.

In~\cite{DBLP:conf/lics/KfouryW95} the intermediate language is an
intrinsically typed $\lambda$-calculus extended with lists, the reduction rule is the
composition of a parallel $\beta$-rule and a commutation rule, and the measure is a
multiset of natural numbers decreasing under the innermost reduction.
In~\cite{DBLP:conf/tlca/Boudol03} the intermediate language is a variant of the
Klop's calculus, extrinsecally intersection typed, and the reduction is the
$\beta$-reduction, but copying its argument when it is going to be erased; the
weak normalization proof is obtained through the weak cut-elimination property
of the type derivation, and the measure is a pair of natural numbers, decreasing
under the strategy choosing the innermost redex of maximum degree.

Since our work uses an approach in some sense similar
to~\cite{DBLP:conf/lics/KfouryW95,DBLP:conf/tlca/Boudol03}, we first highlight
the novelties that we consider important with respect to them.
\begin{itemize}
\item We share with~\cite{DBLP:conf/lics/KfouryW95} the use of an intermediate
language which is intrinsically typed by intersection types, but while their
language is \emph{ad hoc}, and its properties are not explored, $\CalcLamT$ is
shown here to be the Church version of the Curry-style system. Indeed,
$\CalcLame$ and $\CalcLamT$ simulate each other, and reduction corresponds to
cut-elimination in the derivation tree; we consider this to be an important
result on its own.
\item We do not need to define an \emph{ad hoc} reduction rule; $\CalcLamT$ uses
directly the $\beta$-reduction, in its typed version.
\item The codomain of our measure is simpler than multisets and pairs
in~\cite{DBLP:conf/lics/KfouryW95,DBLP:conf/tlca/Boudol03}, consisting of a
single natural number.
\item Our technique provides a measure that decreases for \emph{every} reduction
strategy. There are proofs with these characteristics for the $\lambda$-calculus, such
as Gandy's proof and its derivatives~\cite{gandy80sn,dv87}, but not for
intersection types, as far as we know.
\end{itemize}
\noindent
Despite these differences, their works and ours rely fundamentally on
Turing's remark that contracting a redex cannot create redexes of higher or
equal degree~\cite{gandy80turing,BarendregtM13}, which is the basis for Turing's
proof of weak normalization of the simply typed
$\lambda$-calculus~\cite{gandy80turing}.

The approach of~\cite{DBLP:journals/mscs/BucciarelliPS03} is completely
different: indeed, the normalization of intersection types here is not the main
result, but it is an intermediate step in the proof that computational power of
intersection types is the same as that of simple types. The authors define an
embedding from intersection type to simple type derivations which commutes with
the $\beta$-reduction. So they do not define any decreasing measure. Let us notice
that the embedding translates terms to linear terms, and that the intersection
is used in a non commutative way, since the translation of $\typ \wedge \typtwo$ does
not coincide with that of $\typtwo \wedge \typ$.

One could define a measure for intersection types by combining the translation
with one of the existing measures for simply typed $\lambda$-calculus, for instance
the measure $\meassym$ in \cite{BarenbaumS23}. Due to information lost in the
translation, this would supply a natural number greater or equal than the one
obtained in this paper. Since decreasing measures serve as an upper bound for
the longest reduction chain of terms, a tighter measure provides a more accurate
computational analysis and would therefore be generally preferred.

It is worth noting that even the combination with a simply typed $\lambda$-calculus
exact measure (\eg~\cite{dv87}) would not yield a tighter (nor looser) measure
than $\meassym$. The loss of intersection type information that occurs during
the translation is a hard limitation for any measure based on it; they cannot be
refined to exactness. Our more coupled approach, since it operates \emph{within}
the intersection system, provides a finer analysis of reduction and thus allows
for further improvement, which is indeed one of the proposed future works. We
consider the refinability of our measure an advantage over measures based on
such translation.

\paragraph{Organization of this paper}
Section~\ref{lambdaT} presents the languages $\CalcLame$ and $\CalcLamT$, their
relations and properties. Section~\ref{lambdaTm} constructs a decreasing measure
that entails strong normalization of $\CalcLamT$ by means of an auxiliary memory
calculus $\CalcLamTm$. Finally, Section~\ref{concl} shows how this can
be used to prove strong normalization of $\CalcLame$, concludes, and proposes
future works.

\section{An intrinsically typed presentation of idempotent intersection types}
\label{lambdaT}

This section starts by presenting $\CalcLame$ (\ref{subsec:CalcLame}),
an intersection type assignment system for $\lambda$-terms,
closely related to other well-known systems that characterize strong normalization.
In Section~\ref{subsec:CalcLamT}, we formulate an intrinsically
typed version of the system, called $\CalcLamT$,
and we study its properties, including subject reduction
and confluence.
In Section~\ref{subsec:CalcLame_CalcLamT_correspondence} we establish a formal
proof-theoretical connection between $\CalcLame$ and $\CalcLamT$, and 
we prove that $\CalcLame$ gives type to all the strongly normalizing terms.
Recall that the set of untyped $\lambda$-terms $\CalcLam$ is defined as usual
by the grammar
$\utm ::= \var \mid \lam{\var}{\utm} \mid \app{\utm}{\utm}$.
The sets of
\defnn{types} ($\typ,\typtwo,\hdots$) and
\defnn{set-types} ($\typs,\typstwo,\hdots$)
are given mutually inductively as follows:
\[
  \typ ::= \basetyp \mid \typs \arrow \typ \qquad
  \text{$(\typs \neq \emptyset)$}
  \qquad \qquad \qquad \qquad 
  \typs ::= \set{A_j}_\jJ
\]
where $\basetyp,\basetyptwo,\hdots$ range over a denumerable
set of base types, and $\typs$ stands for a finite set of types\footnote{Note
that $\typs ::= \set{A_j}_\jJ$ is not intended to be a grammatical production
that determines a concrete syntax for finite sets, but rather as a formation
rule meaning that a set-type $\typs$ is given by a finite set of types.},
where $J$ denotes a finite set of indices,
and $\set{\typ_j}_\jJ$ is the set $\set{\typ_j \ST j \in J}$.
When we write a set-type,
we keep the invariant that there are no repetitions,
\ie that if $h \neq k$ then $\typ_h \neq \typ_k$,
so that the set is of cardinality $|J|$.
Sometimes we will use the explicit notation
$\set{A_1,...,A_n}$.
Note that a set-type is a finite set of types that may be empty,
but the domain of an arrow type is always a \emph{non-empty} set-type.

Sometimes we write $\typ\arrow\typtwo$ for $\set{\typ}\arrow\typtwo$.
A typing context ($\tctx,\tctxtwo,\hdots$)
is a function mapping each variable to a set-type
such that $\tctx(\var) \neq \emptyset$ for finitely
many variables~$\var$.
We write $\dom{\tctx}$ for the set of variables which
$\tctx$ assigns to a non-empty set-type.
We adopt the standard notation for typing contexts,
\eg writing $\var_1:\typs_1,\hdots,\var_n:\typs_n$
for the context $\tctx$ such that $\tctx(\var_i) = \typs_i$
for each $\itoN$
and $\tctx(\vartwo) = \emptyset$
for every $\vartwo \notin \set{\var_1,\hdots,\var_n}$.
We write $\tctx\cup\tctxtwo$ for the typing context such that
$(\tctx\cup\tctxtwo)(\var) = \tctx(\var) \cup \tctxtwo(\var)$
for every $\var$,
and we write $\tctx \subseteq \tctxtwo$
if $\tctx(\var) \subseteq \tctxtwo(\var)$
for every $\var$.
Note in particular that
$\tctx \subseteq (\tctx,\var:\typs)$
if $\var \notin \dom{\tctx}$.
If $\judgment$ is a judgement, we write $\derivs{\deriv}{\judgment}$
to mean that $\deriv$ is a particular derivation proving $\judgment$.

\subsection{Extrinsically typed system}
\label{subsec:CalcLame}

We define a Curry-style intersection type assignment system,
which is a variant of well-known systems in the literature,
starting with the one introduced by Coppo and Dezani~\cite{DBLP:journals/ndjfl/CoppoD80}.

\begin{defn}[The $\CalcLame$ type assignment system]
The $\CalcLame$ calculus is defined by means of two forms of judgement:
\begin{enumerate}
\item
  $\judge{\tctx}{\utm}{\typ}$
  meaning that the $\lambda$-term $\utm$ has type $\typ$
  under the context $\tctx$.
\item
  $\judges{\tctx}{\utm}{\typs}$
  meaning that the $\lambda$-term $\utm$
  has set-type $\typs$ under the context $\tctx$.
\end{enumerate}
The typing rules are as follows:
\[
  \infer[\ruleeAx]{
    \judge{\tctx,\var:\typs}{\var}{\typtwo}
  }{
    \typtwo \in \typs
  }
  \qquad
  \infer[\ruleeMulti]{
    \judges{\tctx}{\utmtwo}{\set{\typ_j}_\jJ}
  }{
    (\judge{\tctx}{\utmtwo}{\typ_j})_{\jJ}
    \quad
    (\forall h,k \in J.\,\, h \neq k \implies \typ_h \neq \typ_k)
  }
\]
\[
  \infer[\ruleeArri]{
    \judge{\tctx}{\lam{\var}{\utm}}{\typs\arrow\typtwo}
  }{
    \judge{\tctx,\var:\typs}{\utm}{\typtwo}
  }
  \qquad
  \infer[\ruleeArre]{
    \judge{\tctx}{\app{\utm}{\utmtwo}}{\typtwo}
  }{
    \judge{\tctx}{\utm}{\typs\arrow\typtwo}
    \quad
    \typs\neq\emptyset
    \quad
    \judges{\tctx}{\utmtwo}{\typs}
  }
\]

\noindent Note that we restrict the codomain of contexts so that $\tctx(x) \not = \emptyset$, so rule $\ruleeArri$ cannot introduce an empty set to the left of the arrow.
\end{defn}

\begin{exmp}
\label{ex:typing_deriv_extrinsic}
Let $\tctx = \var:\set{\set{\typ,\typtwo}\to\typthree,\typ,\typtwo}$. Then:
\[
  \infer[\ruleeArri]{
    \judge{\emptyset}{\lam{\var}{\var\,\var}}{\set{\set{\typ,\typtwo}\to\typthree,\typ,\typtwo}\to\typthree}
  }{
    \infer[\ruleeArre]{
      \judge{\tctx}{\var\,\var}{\typthree}
    }{
      \infer[\ruleeAx]{
        \judge{\tctx}{\var}{\set{\typ,\typtwo}\to\typthree}
      }{
        \set{\typ,\typtwo}\to\typthree \in \tctx(\var)
      }
      \quad
      \infer[\ruleeMulti]{
        \judges{\tctx}{\var}{\set{\typ,\typtwo}}
      }{
        \infer[\ruleeAx]{
          \judge{\tctx}{\var}{\typ}
        }{
          \typ \in \tctx(\var)
        }
        \quad
        \infer[\ruleeAx]{
          \judge{\tctx}{\var}{\typtwo}
        }{
          \typtwo \in \tctx(\var)
        }
      }
    }
  }
\]
\end{exmp}

\begin{rem}
The system $\CalcLame$ presented above is not exactly the same system as that
of~\cite{DBLP:journals/ndjfl/CoppoD80}, for two reasons.
First, $\CalcLame$ uses sets instead of intersections, writing
$\set{\typ,\typtwo}\to\typthree$ rather than $(\typ\wedge\typtwo)\to\typthree$.
This can be understood only as a difference in notation, given that
the intersection type constructor of~\cite{DBLP:journals/ndjfl/CoppoD80}
is assumed to be associative, commutative, and idempotent.
Second, $\CalcLame$ allows sets (\ie intersections)
to appear only to the \emph{left} of the arrow ($\to$) connective.
For example,
$(\typ\wedge\typtwo)\to(\typthree\wedge\typfour)$ is a well-formed type in
the system of~\cite{DBLP:journals/ndjfl/CoppoD80},
whereas $\set{\typ,\typtwo}\to\set{\typthree,\typfour}$ is \textbf{not} a
well-formed type in $\CalcLame$.
This means that types in $\CalcLame$ are \emph{strict types},
in the sense introduced by van~Bakel~\cite{Bakel92}.
Despite these two differences, we will prove that 
the original system of~\cite{DBLP:journals/ndjfl/CoppoD80}
and~$\CalcLame$ are equivalent in typability power,
as they both characterize strongly normalizable $\lambda$-terms.
\end{rem}

\subsection{Intrinsically typed system}
\label{subsec:CalcLamT}

We define a Church-style presentation of the system,
inspired by the linearization
proposed by Kfoury \cite{Kfoury00}.

\begin{defn}[The $\CalcLamT$ type assignment system]
\label{def_CalcLamT}
The sets of \defnn{terms} ($\tm,\tmtwo,\hdots$)
and \defnn{set-terms} ($\tms,\tmstwo,\hdots$)
are given by the following grammar:
\[
  \tm ::= \var^\typ
  \mid \lam{\var^{\typs}}{\tm}
  \mid \tm\,\tms
  \qquad \qquad \qquad 
  \tms ::= \set{\tm_j}_\jJ \qquad
\]
where $J$ stands for a finite set of indices
and $\set{\tm_j}_\jJ$ is the set $\set{\tm_j \ST j \in J}$.
Note that set-terms are finite sets of terms.
When we write a set-term,
we keep the invariant that there are no repetitions,
\ie that if $h \neq k$ then $\tm_h \neq \tm_k$.
The $\CalcLamT$ calculus is defined by means of two forms of judgement:
\begin{enumerate}
\item
  $\judgT{\tctx}{\tm}{\typ}$
  meaning that a term $\tm$ is of type $\typ$
  under the context $\tctx$.
\item
  $\judgTs{\tctx}{\tms}{\typs}$
  meaning that a set-term $\tms$
  is of set-type $\typs$
  under the context $\tctx$.
\end{enumerate}
The typing rules are as follows:
\[
  \infer[\ruleTAx]{
    \judgT{\tctx, \var:\typs}{\var^\typtwo}{\typtwo}
    }{\typtwo \in \typs }
  {
  }
  \qquad
  \infer[\ruleTMulti]{
    \judgTs{\tctx}{\set{\tm_j}_\jJ}{\set{\typ_j}_\jJ}
  }{
    (\judgT{\tctx}{\tm_j}{\typ_j})_{\jJ}
    \quad
      (\forall h,k \in J.\, h\neq k \implies \typ_h \neq \typ_k)
  }
\]
\[
  \infer[\ruleTArri]{
    \judgT{\tctx}{\lam{\var^\typs}{\tm}}{\typs\arrow\typtwo}
  }{
    \judgT{\tctx,\var:\typs}{\tm}{\typtwo}
  }
  \qquad
  \infer[\ruleTArre]{
    \judgT{\tctx}{\app{\tm}{\tmstwo}}{\typtwo}
  }{
    \judgT{\tctx}{\tm}{\typs\arrow\typtwo}
    \qquad
    \judgTs{\tctx}{\tmstwo}{\typs}
    \qquad
    \typs \neq \emptyset
  }
\]
A term $\tm$ is \defnn{typable} if $\judgT{\tctx}{\tm}{\typ}$
holds for some $\tctx,\typ$.
Similarly, a set-term is typable if $\judgTs{\tctx}{\tms}{\typs}$
holds for some $\tctx,\typs$.
Unless otherwise specified, when we speak
of {\em term} (resp. {\em set-term})
we mean {\em typable term} (resp. {\em typable set-term}).
We write $\CalcLamT$ for the set of typable terms,
\ie
$\CalcLamT \eqdef \set{\tm \ST \text{$\tm$ is typable}}$.
The notions of free and bound occurrences of variables
are defined as usual.
The set of \defnn{free variables of type $\typ$} of a term $\tm$
is written $\fvtyp{\typ}{\tm}$
and defined as the set of variables $\var$
such that $\var^\typ$ occurs free in $\tm$.
For example,
$\fvtyp{\typ}{\lam{\vartwo^{\set{\set{\typ,\typtwo}\arrow\typthree}}}{\vartwo^{\set{\typ,\typtwo}\arrow\typthree}\set{\var^\typ,\varthree^\typtwo}}} = \set{\var}$
when $\typ\neq\typtwo$.
The set of \defnn{free variables} of a term $\tm$
is written $\fv{\tm}$ .
Terms are considered up to $\alpha$-renaming of bound variables.
\end{defn}

\begin{notation}
Sometimes we write $\app{\tm}{\tmtwo}$ to stand for $\app{\tm}{\set{\tmtwo}}$.  We may omit type annotations over variables if they are clear from the context.
\end{notation}

\begin{rem}
  \label{rem:inj}
Note that, in the rule $\ruleTMulti$,
we explicitly require that the function $j \mapsto \typ_j$
is injective
(\ie $\forall h,k \in J.\, h\neq k \implies \typ_h \neq \typ_k$).
For example,
let $\tctx = (\var:\set{\typ,\typtwo},\vartwo:\set{\typ})$.
Then the following are \textbf{valid} judgments:
\[
  \judgTs{\tctx}{\set{\var^\typ}}{\set{\typ}}
  \qquad
  \judgTs{\tctx}{\set{\var^\typ,\var^\typtwo}}{\set{\typ,\typtwo}}
  \qquad
  \judgTs{\tctx}{\set{\var^\typtwo,\vartwo^\typ}}{\set{\typ,\typtwo}}
\]
On the other hand, the following judgments are \textbf{not valid}:
\[
  \judgTs{\tctx}{\set{\var^\typ,\vartwo^\typ}}{\set{\typ}}
  \qquad
  \judgTs{\tctx}{\set{\var^\typ,\var^\typ}}{\set{\typ}}
\]
\end{rem}

\begin{exmp}
The derivation of Example~\ref{ex:typing_deriv_extrinsic} can be encoded in
$\CalcLamT$ as the following term, which has type $\set{\set{\typ,\typtwo}\to\typthree,\typ,\typtwo}\to\typthree$ in
the empty context:
\[
  \lam{\var^{\set{\set{\typ,\typtwo}\to\typthree,\typ,\typtwo}}}{
    \var^{\set{\typ,\typtwo}\to\typthree}
    \,
    \set{
      \var^{\typ}
    , \var^{\typtwo}
    }
  }
\]
\end{exmp}

\begin{rem}
Note that, in general, a set-term may be an \emph{empty} set of terms,
but the argument of an application is always a \emph{non-empty}
set-term.
We let set-types and set-terms to be empty
to be able to write expressions such as
$\tms = \set{\tm} \cup \tmstwo$
to mean that $\tms$ is a set containing at least one element $\tm$,
leaving the possibility that $\tmstwo$ may be empty.
\end{rem}

\begin{rem}[Subterm property]
If $\tm$ is typable, all of its subterms are typable.
\end{rem}

\begin{lem}[Weakening and strengthening]\label{lem:weak}
\quad
\begin{enumerate}
\item
  If $\judgT{\tctx}{\tm}{\typ}$ (resp. $\judgTs{\tctx}{\tms}{\typs}$)
  and $\tctx \subseteq \tctxtwo$
  then $\judgT{\tctxtwo}{\tm}{\typ}$ (resp. $\judgTs{\tctxtwo}{\tms}{\typs}$).
\item
  If $\judgT{\tctx,\var:\set{\typ}\cup\typstwo}{\tm}{\typthree}$ (resp. $\judgTs{\tctx,\var:\set{\typ}\cup\typstwo}{\tms}{\typsthree}$)
  and $\var \notin \fvtyp{\typ}{\tm}$,
  then $\judgT{\tctx'}{\tm}{\typthree}$ (resp. $\judgTs{\tctx'}{\tms}{\typsthree}$), where $\tctx'$ is $(\tctx,\var:\typstwo)$ if $\typstwo \not = \emptyset$, and $\tctx$ otherwise.


\end{enumerate}
\end{lem}

\begin{defn}[Minimal typing context]
For each term $\tm$ (resp. set-term $\tms$),
its \emph{minimal typing context}
is written $\MinCtx{\tm}$ (resp. $\MinCtx{\tms}$)
and defined as follows:
\[
  \begin{array}{lll@{\qquad}lll}
    \MinCtx{\var^\typ}
    & \eqdef &
    \var:\set{\typ}
  &
    \MinCtx{\lam{\var^\typs}{\tm}}
    & \eqdef &
    \MinCtx{\tm} \setminus \var
  \\
    \MinCtx{\tm\,\tmstwo}
    & \eqdef &
    \MinCtx{\tm} \cup \MinCtx{\tmstwo}
  &
    \MinCtx{\set{\tm_1,\hdots,\tm_n}}
    & \eqdef &
    \cup_{i=1}^n \MinCtx{\tm_i}
  \end{array}
\]
where $\tctx \setminus \var$
is defined in such a way that
$(\tctx\setminus\var)(\var) = \emptyset$
and $(\tctx\setminus\var)(\vartwo) = \tctx(\vartwo)$
for $\var \neq \vartwo$.
\end{defn}

\begin{lem}[Type uniqueness]
\label{type_uniqueness}
If $\judgT{\tctx_1}{\tm}{\typ_1}$
and $\judgT{\tctx_2}{\tm}{\typ_2}$
then $\typ_1 = \typ_2$.
Furthermore, $\MinCtx{\tm} \subseteq \tctx_1$
and $\MinCtx{\tm} \subseteq \tctx_2$
and $\judgT{\MinCtx{\tm}}{\tm}{\typ_1}$.
\end{lem}

\begin{rem}
Sometimes we write $\tm^\typ$ to mean
that $\judgT{\MinCtx{\tm}}{\tm}{\typ}$.
Note that, by the previous lemma,
there is a unique $\typ$ such that $\tm^\typ$.
\end{rem}

\begin{obs}[Bijection between a set-term and its set-type]\label{obs:bi}

  Let $\judgTs{\tctx}{\tmstwo}{\typs}$.
  Given $\typ' \in \typs$, we know by rule $\ruleTMulti$ that there exists a derivation $\judgT\tctx{\tmtwo'}{\typ'}$. This is the only derivation with type $\typ'$ in the premises of the rule: otherwise, we would contradict the rule's side condition by having different derivations with the same type.
  Therefore, $\tmtwo'$ is unique.
  Given $\tmtwo' \in \tmstwo$, we know by rule $\ruleTMulti$ that it has a type $\typ'$; by Lemma~\ref{type_uniqueness}, this type is unique.
   


\end{obs}

\begin{defn}[Substitution in $\CalcLamT$]
  We define operations of capture-avoiding substitution for terms and
  set-terms by mutual recursion.
  Let $\judgTs{\tctx}{\tmstwo}{\typs}$
    with $\tmstwo = \set{\tmtwo_i}_\iI$, $\typs = \set{\typ_i}_\iI$, and
    $\judgT{\tctx}{\tmtwo_i}{\typ_i}$ for all $i \in I$.
we define $\tm\sub{\var^\typs}{\tmstwo}$ recursively as follows.

  \[
    \begin{array}{rcl}
      \var^{\typ_i}\sub{\var^\typs}{\tmstwo}
    & \eqdef &
      \tmtwo_i
    \\
      \vartwo^\typtwo\sub{\var^\typs}{\tmstwo}
    & \eqdef &
      \vartwo^\typtwo
    \\
      (\lam{\vartwo^\typstwo}{\tm})\sub{\var^\typs}{\tmstwo}
    & \eqdef &
      \lam{\vartwo^\typstwo}{\tm\sub{\var^\typs}{\tmstwo}}
    \\
      (\app{\tm}{\tmsthree})\sub{\var^\typs}{\tmstwo}
    & \eqdef &
               \app{\tm\sub{\var^\typs}{\tmstwo}}{\tmsthree\sub{\var^\typs}{\tmstwo}} \\
    \set{\tm_1,\hdots,\tm_n}\sub{\var^\typs}{\tmstwo} &
    \eqdef &
   \set{\tm_1\sub{\var^\typs}{\tmstwo},\hdots,\tm_n\sub{\var^\typs}{\tmstwo}}
    \end{array}
  \]
\noindent
Note that, in the first case, $\tmtwo_i$ exists and is unique due to
Observation~\ref{obs:bi}.
\end{defn}
\begin{lem}[Substitution Lemma]
  \label{lem:sub}
 Let $\judgT{\tctx,\var:\typs}{\tm}{\typtwo}$ and $\judgTs{\tctx}{\tmstwo}{\typs}$.
 Then $\judgT{\tctx}{\tm\sub{\var^\typs}{\tmstwo}}{\typtwo}$.
\end{lem}
\begin{proof}
 By induction on $\tm$, generalizing the statement to set-terms, where
    the interesting cases are:
  \begin{enumerate}
  \item If $\tm = \var^{\typtwo}$, then $\judgT{\tctx,x:\typs}{\var^{\typtwo}}\typtwo$
    with $\typtwo\in\typs$. By Observation~\ref{obs:bi}, there exists a unique $\tmtwo_i \in
    \tmstwo$ such that $\judgT\tctx{\tmtwo_i}\typtwo$. Therefore,
    $\var^{\typtwo}\sub{\var^\typs}{\tmstwo} = \tmtwo_i$ of type $\typtwo$.
  \item If $\tms = \set{\tm_1,\hdots,\tm_n}$, then $\typtwo$ is indeed a set-type
    $\typstwo = \set{\typtwo_1,\hdots,\typtwo_n}$, and
    $\judgTs{\tctx,x:\typs}{\set{\tm_1,\hdots,\tm_n}}
    {\set{\typtwo_1,\hdots,\typtwo_n}}$. By IH,
    $\judgT{\tctx}{\tm_i\sub{\var^\typs}{\tmstwo}}{\typtwo_i}$ for all $i$ such that $1 \leq i \leq n$, so by $\ruleTMulti$ we have
    $\judgTs{\tctx}{\set{\tm_1\sub{\var^\typs}{\tmstwo},\hdots,\tm_n
        \sub{\var^\typs}{\tmstwo}}}{\set{\typtwo_1,\hdots,\typtwo_n}}$.
    Note that substitution preserves cardinality in typed set-term derivations. If we had $\tm_i\sub{\var^\typs}\tmstwo = \tm_j\sub{\var^\typs}\tmstwo$ with distinct $i,j$ such that $1 \leq i,j \leq n$, they would have the same type by Lemma \ref{type_uniqueness}. This would imply that $\tm_i$ and $\tm_j$ have the same type, which is impossible by rule $\ruleTMulti$.
  \end{enumerate}
\end{proof}

\begin{defn}[Reduction in $\CalcLamT$]
A notion of reduction called $\symT$-reduction is
defined over terms and set-terms by the following rule,
closed by congruence under arbitrary contexts:
\[
  (\lam{\var^\typs}{\tm})\,\tmstwo
  \toT
  \tm\sub{\var^\typs}{\tmstwo}
\]
Note in particular that congruence rules allow reducing inside
the argument of an application,
\ie
$\tmstwo \toT \tmsthree$
implies
$\app{\tm}{\tmstwo} \toT \app{\tm}{\tmsthree}$,
and inside any of the elements of a set-term,
\ie $\tm_i \toT \tm_i'$ implies $\set{\tm_1,\hdots,\tm_i,\hdots,\tm_n} \toT \set{\tm_1,\hdots,\tm_i',\hdots,\tm_n}$.
We write $\toTs$ for the reflexive--transitive closure of $\toT$.
\end{defn}
Reduction in $\CalcLamT$ enjoys the following properties.
See Appendices~\ref{appendix:section:CalcLamT_subject_reduction} and~\ref{appendix:section:CalcLamT_confluence} for detailed proofs.

\begin{restatable}[Subject reduction]{prop}{subjectreduction}
  \label{prop:sub-red}
If $\tm \toT \tm'$ and $\judgT{\tctx}{\tm}{\typ}$
then $\judgT{\tctx}{\tm'}{\typ}$.
\end{restatable}

\begin{restatable}[Confluence]{prop}{confluenceprop}
  \label{prop:confluence}
  If $\tm \toTs \tmtwo$ and $\tm \toTs \tmthree$ then there is $\tmfour$ such
  that $\tmtwo \toTs \tmfour$ and $\tmthree \toTs \tmfour$.
\end{restatable}

\subsection{Correspondence between Curry and Church-style systems}
\label{subsec:CalcLame_CalcLamT_correspondence}

A typable term in the Church-style system $\CalcLamT$
does not necessarily represent
a derivation tree in the Curry-style system $\CalcLame$.
To relate these systems, we consider the following notion:

\begin{defn}[Refinement]
We define a binary relation $\judgRefine{\tm}{\utm}$
between typable terms $\tm \in \CalcLamT$
and untyped $\lambda$-terms $\utm \in \CalcLam$,
as well as a binary relation $\judgRefine{\tms}{\utm}$
between typable set-terms $\tms \subset \CalcLamT$
and untyped $\lambda$-terms $\utm \in \CalcLam$,
according to the inductive rules below.
In that case, we say that $\tm$ is a \defnn{refinement}
of $\utm$.
\[
  \infer[]{
    \judgRefine{\var^\typ}{\var}
  }{
  }
  \qquad
  \infer[]{
    \judgRefine{\lam{\var^\typ}{\tm}}{\lam{\var}{\utm}}
  }{
    \judgRefine{\tm}{\utm}
  }
  \qquad
  \infer[]{
    \judgRefine{\app{\tm}{\tmstwo}}{\app{\utm}{\utmtwo}}
  }{
    \judgRefine{\tm}{\utm}
    &
    \judgRefine{\tmstwo}{\utmtwo}
  }
  \qquad
  \infer[]{
    \judgRefine{\set{\tm_1,\hdots,\tm_n}}{\utm}
  }{
    (\judgRefine{\tm_i}{\utm})_{\itoN}
    &
    n > 0
  }
\]
\end{defn}

\begin{obs} \label{obs:=}\quad
\begin{enumerate}
\item \label{uno} If $\judgRefine{\tm}{\utm_1}$ and $\judgRefine{\tm}{\utm_2}$ then $\utm_1 = \utm_2$.
\item \label{due} If $\judgRefine{\tm}{\utm}$ and $\judgRefine{\tmstwo}{\utmtwo}$ then $\judgRefine{\tm\sub{\var^\typs}{\tmstwo}}{\utm\sub{\var}{\utmtwo}}$.
\end{enumerate}
\end{obs}

\begin{defn}[Uniformity and type erasure]
A typable term $\tm \in \CalcLamT$ is \defnn{uniform} if there exists $\utm \in \CalcLam$
such that $\judgRefine{\tm}{\utm}$.
Similarly, a typable set-term $\tms \subseteq \CalcLamT$
is uniform if there exists $\utm \in \CalcLam$
such that $\judgRefine{\tms}{\utm}$.

If $\tm \in \CalcLamT$ is a uniform term,
its \defnn{type erasure} $\typerase{\tm} \in \CalcLam$ is defined
as the unique $\lambda$-term such that $\judgRefine{\tm}{\typerase{\tm}}$.
Similarly,
if $\tms \subset \CalcLamT$ is a uniform set-term,
$\typerase{\tms} \in \CalcLam$ is defined as the unique $\lambda$-term
such that $\judgRefine{\tms}{\typerase{\tms}}$.
This can be written as a recursive definition:
\[
  \typerase{(\var^\typ)}
  \eqdef
  \var
  \qquad
  \typerase{(\lam{\var^\typs}{\tm})}
  \eqdef
  \lam{\var}{\typerase{\tm}}
  \qquad
  \typerase{(\app{\tm}{\tmstwo})}
  \eqdef
  \app{\typerase{\tm}}{\typerase{\tmstwo}}
  \qquad
  \typerase{\set{\tmtwo_1,\hdots,\tmtwo_n}}
  \eqdef
  \typerase{\tmtwo_1}
\]
If $\tm \in \CalcLamT$ is uniform,
for each derivation $\derivs{\deriv}{\judgT{\tctx}{\tm}{\typ}}$,
we construct a derivation
$\derivs{\typerase{\deriv}}{\judge{\tctx}{\typerase{\tm}}{\typ}}$
by erasing all the type annotations in terms.
Similarly, each derivation $\derivs{\deriv}{\judgTs{\tctx}{\tms}{\typs}}$
is mapped to a derivation $\derivs{\typerase{\deriv}}{\judges{\tctx}{\typerase{\tms}}{\typs}}$.
\end{defn}

\begin{defn}[Type decoration]
For each derivation $\derivs{\deriv}{\judge{\tctx}{\utm}{\typ}}$,
we construct a term $\tm \in \CalcLamT$
and a derivation
$\derivs{\typdec{\deriv}}{\judgT{\tctx}{\tm}{\typ}}$ inductively as below.
Similarly,
for each derivation $\derivs{\deriv}{\judges{\tctx}{\utm}{\typs}}$
we construct a set-term $\tms \in \CalcLamT$
and a derivation $\derivs{\typdec{\deriv}}{\judgTs{\tctx}{\tms}{\typs}}$.
\[
  \begin{array}[t]{llll}
    \infer[\ruleeAx]{
      \judge{\tctx,\var:\typs}{\var}{\typtwo}
    }{
      \typtwo \in \typs
    }
  & \rightsquigarrow &
    \infer[\ruleTAx]{
      \judgT{\tctx,\var:\typs}{\var^\typtwo}{\typtwo}
    }{\typtwo \in \typs
    }
  \\
  \\
    \infer[\ruleeArri]{
      \judge{\tctx}{\lam{\var}{\utm}}{\typs\arrow\typtwo}
    }{
      \judge{\tctx,\var:\typs}{\utm}{\typtwo}
    }
  & \rightsquigarrow &
    \infer[\ruleTArri]{
      \judgT{\tctx}{\lam{\var^\typs}{\tm}}{\typs\arrow\typtwo}
    }{
      \judgT{\tctx,\var:\typs}{\tm}{\typtwo}
    }
  \\
  \\
    \infer[\ruleeArre]{
      \judge{\tctx}{\app{\utm}{\utmtwo}}{\typtwo}
    }{
      \judge{\tctx}{\utm}{\typs\arrow\typtwo}
      \quad
      \judges{\tctx}{\utmtwo}{\typs}
    }
  &
    \rightsquigarrow
  &
    \infer[\ruleTArre]{
      \judgT{\tctx}{\app{\tm}{\tmstwo}}{\typtwo}
    }{
      \judgT{\tctx}{\tm}{\typs\arrow\typtwo}
      \quad
      \judgTs{\tctx}{\tmstwo}{\typs}
    }
  \\
  \\
    \infer[\ruleeMulti]{
      \judges{\tctx}{\utmtwo}{\set{\typ_1,\hdots,\typ_n}}
    }{
      (\judge{\tctx}{\utmtwo}{\typ_i})_{\itoN}
    }
  &
    \rightsquigarrow
  &
    \infer[\ruleTMulti]{
      \judgTs{\tctx}{\set{\tmtwo_1,\hdots,\tmtwo_n}}{\set{\typ_1,\hdots,\typ_n}}
    }{
      (\judgT{\tctx}{\tmtwo_i}{\typ_i})_{\itoN}
    }
  \end{array}
\]
\end{defn}

\begin{thm}[Correspondence]
\label{lami_lamT_correspondence}
\quad
\begin{enumerate}
\item \label{corr-one}
  If $\derivs{\deriv}{\judge{\tctx}{\utm}{\typ}}$, then
  there exists a uniform $\tm \in \CalcLamT$ such that
  $\derivs{\typdec{\deriv}}{\judgT{\tctx}{\tm}{\typ}}$
  and $\judgRefine{\tm}{\utm}$.
  Moreover, $\typerase{(\typdec{\deriv})} = \deriv$.
\item \label{corr-two}
  If $\derivs{\deriv}{\judgT{\tctx}{\tm}{\typ}}$
  and $\tm$ is uniform
  then $\derivs{\typerase{\deriv}}{\judge{\tctx}{\typerase{\tm}}{\typ}}$.
  Moreover, $\typdec{(\typerase{\deriv})} = \deriv$.
\end{enumerate}
\end{thm}
\begin{proof}
\quad
\begin{enumerate}
\item We generalize the statement to set-terms, \ie if $\derivs{\deriv}{\judges{\tctx}{\utm}{\typs}}$ then there exists a uniform $\tms \in \CalcLamT $ such that $\derivs{\typdec{\deriv}}{\judgTs{\tctx}{\tms}{\typs}} $ and $\judgRefine{\tms}{\utm}$.
We proceed by induction on the derivation $\deriv$. The only not immediate case is when $\derivs{\deriv}{\judges{\tctx}{\utm}{\typs}} $. Then this judgement has been obtained by rule $\ruleeMulti$ with premises $(\judge{\tctx}{\utm}{\typ_i})_{\itoN}$; by \ih there are derivations proving $\judges{\tctx}{\tm_i}{\typ_i}$, where $\judgRefine{\tm_i}{\utm}$. Then, by definition, $\judgRefine{\set{\tm_1,\hdots,\tm_n}}{\utm}$, and the proof follows by rule $\ruleTMulti$.
The other cases come directly from the definition of type decoration. 
\item We generalize the statement to set-terms, \ie if
  $\derivs{\deriv}{\judgTs{\tctx}{\tms}{\typs}}$ and $\tms$ is uniform then
  $\derivs{\typerase{\deriv}}{\judges{\tctx}{\typerase{\tms}}{\typs}}$.
  We proceed by induction on the derivation $\deriv$. The only not immediate case is when $\derivs{\deriv}{\judgTs{\tctx}{\tms}{\typs}} $. Let $\tms=\set{\tm_1,\hdots,\tm_n}$, then $\judgRefine{\tm}{\utm}$ implies, by definition of uniformity, $\judgRefine{\tm_i}{\utm}$. The judgement has been obtained by rule $\ruleTMulti$, with premises ${\judgT{\tctx}{\tm_i}{\typ_i}} $, so by \ih, $\judge{\tctx}{\typerase{\tm_i}}{\typ_i}$ where 
$\judgRefine{\tm_i}{\typerase{\tm_i}}$. By Observation~\ref{obs:=}.\ref{uno}, $\judgRefine{\tm_i}{\typerase{\tm_i}}$ and $\judgRefine{\tm_i}{\utm}$ imply $\typerase{\tm_i}=\utm$, so $\judge{\tctx}{\tm}{\typ_i}$ and the result follows by applying rule $\ruleeMulti$. 
The other cases come directly from the definition of type erasure. \qedhere
\end{enumerate}
\end{proof}

\begin{cor}
Given a term $\utm \in \CalcLam$, the following are equivalent:
\begin{enumerate}
\item There exist $\tctx,\typ$ such that $\judge{\tctx}{\utm}{\typ}$ holds.
\item There exists a term $\tm \in \CalcLamT$
      such that $\judgRefine{\tm}{\utm}$.
\end{enumerate}
\end{cor}

Furthermore, $\tobeta$ and $\toT$ simulate each other:

\begin{thm}[Simulation]
\label{lami_lamT_simulation}
\quad
\begin{enumerate}
\item
  If $\utm \tobeta \utmtwo$
  and $\judgRefine{\tm}{\utm}$,
  there exists $\tmtwo \in \CalcLamT$
  such that $\tm \toT^+ \tmtwo$
  and $\judgRefine{\tmtwo}{\utmtwo}$.
\item
  If $\tm \toT \tmtwo$
  and $\judgRefine{\tm}{\utm}$,
  there exist $\utmtwo \in \CalcLam$ and $\tmtwo' \in \CalcLamT$
  such that
  $\utm \tobeta \utmtwo$
  and $\tmtwo \toT^* \tmtwo'$
  and $\judgRefine{\tmtwo'}{\utmtwo}$.
\end{enumerate}
Graphically:
\[
  \xymatrix@R=0cm@C=1.5cm{
    \utm \ar_(.8){\beta}[r] & \utmtwo \\
    \rotatebox{90}{$\sqsubset$} & \rotatebox{90}{$\sqsubset$} \\
    \tm \ar@{.>}^(.8){+}_(.8){\LamTsymbol}[r] & \tmtwo
  }
  \qquad
  \xymatrix@R=0cm@C=1.5cm{
    \tm \ar_(.8){\LamTsymbol}[r] & \tmtwo \ar@{.>}_(.8){\LamTsymbol}^(.8){*}[r] & \tmtwo' \\
    \rotatebox{90}{$\sqsupset$} & & \rotatebox{90}{$\sqsupset$} \\
    \utm \ar@{.>}_(.875){\beta}[rr] & & \utmtwo \\
  }
\]
\end{thm}
\begin{proof}\quad
  \begin{enumerate}
  \item 
    We proceed by induction on $\utm$. The case $\utm=\var$ cannot happen, in case
    $\utm=\lam{x}{\utmtwo}$ the proof follows by \ih.
    Let $\utm =\app{\utmthree}{\utmfour}$, so $\tm = \tmthree \tmsfour $, where
    $\judgRefine{ \tmthree}{\utmthree}$ and $\judgRefine{ \tmsfour}{\utmfour}$.
    The case $\utmthree \tobeta \utmthree'$ is easy, by \ih. Let $\utmfour
    \tobeta \utmfour'$, and let $\tmsfour =\set{\tmfour_1,...,\tmfour_n}$,
    where, by definition, $\judgRefine{\tmfour_i }{\utmfour}$, for all $1\leq i
    \leq n$. By \ih, there are $\tmfour'_i$ such that $\tmfour_i \toT^+
    \tmfour'_i $, such that $\judgRefine{\tmfour'_i}{\utmfour'}$. So
    $\judgRefine{\set{\tmfour'_1,...,\tmfour'_n}} {\utmfour'}$ and $\tmtwo =
    \tmthree\set{\tmfour'_1,...,\tmfour'_n}$. Let
    $\utmthree=\lam{\var}{\utmthree'}$, so $\tm=(\lam{\var^\typ}{\tm'})\tmsfour
    $, where $\judgRefine{\tm' }{\utmthree'}$, and let $\utm \tobeta
    \utmthree'\sub{\var}{\utmfour}$. By Observation~\ref{obs:=}.\ref{due},
    $\judgRefine{\tm'\sub{\var^\typ}{\tmsfour}}{\utmthree'\sub{\var}{\utmfour}}$,
    and since $\tm \toT^+ \tm'\sub{\var^\typ}{\tmsfour}$, $\tmtwo=
    \tm'\sub{\var^\typ}{\tmsfour}$.

\item
  We generalize the statement to set-terms, \ie if $\tms \toT
  \tmstwo$ and $\judgRefine{\tms}{\utm}$, there exist $\utmtwo$ and $\tmstwo'$
  such that $\utm \tobeta \utmtwo$ and $\tmstwo \toTs \tmstwo'$ and
  $\judgRefine{\tmstwo'}{\utmtwo}$. The proof is by induction on terms. Let
  $\tms=\set{\tm_1,\hdots,\tm_n,\tmthree}$, and
  $\tmstwo=\set{\tm_1,\hdots,\tm_n,\tmthree'}$, where $\tmthree \toT \tmthree'$.
  By \ih, there are $\utmtwo$ and $\tmthree''$ such that $\utm \tobeta \utmtwo$,
  $\tmthree' \toT^* \tmthree''$ and $\judgRefine{\tmthree''}{\utmtwo}$. Note
  that $\judgRefine{\tm_i }{\utm}$ for all $\itoN$. Now let $\itoN$. By point~1,
  $\utm \tobeta \utmtwo$ and $\judgRefine{\tm_i }{\utm}$ imply there is $\tm_i'$
  such that $\tm_i \toT^*\tm_i'$ and $\judgRefine{\tm_i'}{\utmtwo}$; then it
  suffices to take $\tmstwo' =\set{\tm'_1,...,\tm'_n,\tmthree''}$. The other
  cases follow by \ih, using Observation~\ref{obs:=}.\ref{due} in case $\tm=
  \app{(\lam{\var^\typ}{\tm'})}{\tmsthree}$, and $\tmtwo = \tm'
  \sub{\var^\typ}{\tmsthree}$. \qedhere
\end{enumerate}
\end{proof}

Some comments are in order.
If $\utmtwo$ is a redex subterm of $\utm$, then
a derivation tree for $\utm$ in $\CalcLame$ may contain more than
one subderivation with subject $\utmtwo$.
Contracting $\utmtwo$ with usual $\beta$-reduction corresponds to
reducing in parallel {\it all} the occurrences of $\utmtwo$ in all the subderivations.
In $\CalcLamT$, this corresponds instead to performing the reductions one at a time.

\begin{exmp}
Let $\utm = \app{(\lam{\var}{\app{\var}{\var}})}{(I\,I)}$,
$\typ = \typtwo\arrow\typtwo$, and $I = \lam{\var}{\var}$.
Consider the derivation $\deriv$:
\[
  \infer[\ruleeArre]{
    \judge{}{\utm}{\typ}
  }{
    \infer[\ruleeArre]{
      \judge{}{\lam{\var}{\app{\var}{\var}}}{\set{\typ\arrow\typ,\typ}\arrow\typ}
    }{
      \infer[]{
        \judge{\var:\set{\typ\arrow\typ,\typ}}{\app{\var}{\var}}{\typ}
      }{
        \vdots
      }
    }
    &
    \infer[\ruleeMulti]{
      \judges{}{I\,I}{\set{\typ\arrow\typ,\typ}}
    }{
      \infer[]{
        \judge{}{I\,I}{\typ\arrow\typ}
      }{
        \vdots
      }
      &
      \infer[]{
        \judge{}{I\,I}{\typ}
      }{
        \vdots
      }
    }{
    }
  }
\]
Then we have that
  $\derivs{\typdec{\deriv}}{\judgT{\ }{\tm}{\typ}}$,
where
  $\tm =
    \app{
      (\lam{\var^{\set{\typ\arrow\typ,\typ}}}{
         \app{\var}{\var}
      })
    }{
      \set{
        (\app{I}{I})^{\typ\arrow\typ},
        (\app{I}{I})^{\typ}
      }
    }$.
Note that:
\begin{enumerate}
\item
  $\tm$ is a uniform term, and in particular $\judgRefine{\tm}{\utm}$.
\item
  $\tm \toT
   \tm_1 =
   \app{
       (\lam{\var^{\set{\typ\arrow\typ,\typ}}}{
          \app{\var}{\var}
       })
     }{
       \set{
         I^{\typ\arrow\typ},
         (\app{I}{I})^{\typ}
       }
     }$,
  where
  $\tm_1$ is not a uniform term.

\item
  $\tm \toT
   \tm_3 = \app{(\app{I}{I})^{\typ\arrow\typ}}{(\app{I}{I})^{\typ}}$
  where $\judgRefine{\tm_3}{\app{I}{I}(\app{I}{I})}$
  and $\utm \tobeta \app{I}{I}(\app{I}{I})$.
  \item
   $\tm_1 \toTs
   \tm_2 =
   \app{
       (\lam{\var^{\set{\typ\arrow\typ,\typ}}}{
          \app{\var}{\var}
       })
     }{
       \set{
         I^{\typ\arrow\typ},
         I^{\typ}
       }
     }$,
  where
  $\tm_2$ is a uniform term; in particular $\judgRefine{\tm_2}{\app{(\lam{\var}{\app{\var}{\var}})}{I}}$
  and $\utm \tobeta \app{(\lam{\var}{\app{\var}{\var})}{I}}$.
\end{enumerate}
\end{exmp}

\begin{lem}[Head subject expansion]\label{lem:inv-sub}
If $\judgT{\tctx}{\tm\sub{\var^\typs}{\tmstwo} \tmstwo_1...\tmstwo_n}{\typtwo}$ and $\judgTs{\tctx}{\tmstwo}{\typs}$, then $\judgT{\tctx}{(\lam {\var^\typs}{\tm} )\,\tmstwo \, \tmstwo_1...\tmstwo_n}{\typtwo}$.
\end{lem}
\begin{proof}
By induction on $n$.
Let $n=0$, so $\judgT{\tctx}{\tm\sub{\var^\typs}{\tmstwo}}{\typtwo}$.
We continue by induction on $\tm$.
The interesting case is that of $\tm=\vartwo^\typtwo$, where $\typtwo \in \typs$.
Then $\judgT{\tctx, \var^\typs}{\var^\typtwo}{\typtwo}$, by rule ($\ruleTAx$), and $\judgT{\tctx}{(\lam {\var^\typs} {\var^\typtwo})\,\tmstwo}{\typtwo}$
by rules ($\ruleTArri$) and ($\ruleTArre$).
The remaining cases are straightforward, by IH and substitution definition.
Let $n>0$. Let $\judgT{\tctx}{\tm\sub{\var^\typs}{\tmstwo} \tmstwo_1...\tmstwo_n \, \tmsthree }{\typtwo}$ and $\judgTs{\tctx}{\tmstwo}{\typs}$. Then, by the rules of the system, $\judgT{\tctx}{\tm\sub{\var^\typs}{\tmstwo} \tmstwo_1...\tmstwo_n \, \tmsthree }{\typtwo}$ comes from
$\judgT{\tctx}{\tm\sub{\var^\typs}{\tmstwo} \tmstwo_1...\tmstwo_n }{\typsthree \arrow\typtwo}$ and $\judgTs{\tctx}{\tmsthree}{\typsthree}$, for some $\typsthree$. By \ih, $\judgT{\tctx}{(\lam {\var^\typ}{\tm} )\,\tmstwo \, \tmstwo_1...\tmstwo_n}{\typsthree \arrow\typtwo}$, and the result follows by rule ($\ruleTArre$).
\end{proof}

\begin{lem}[Strong Normalization typability]
  \label{lem:strong}
  Let $\utm\in \CalcLam$ be strongly normalizing. Then there is $\tm \in \CalcLamT$ such that
  $\judgRefine{\tm}{\utm}$ and $\judgT{\tctx}{\tm}{\typ}$,
  for some $\tctx$ and $\typ$.
 \end{lem}

\begin{proof}
The proof relies on the following inductive definition of the strongly normalizing terms (SN).

\[
\infer[SN1]{\var \utm_1...\utm_n \in SN}{\utm_i \in SN \quad 1 \leq i\leq n} \quad \quad \quad
\infer[SN2]{\lam {\var} {\utm} \in SN}{\utm \in SN}
\]

\[
\infer[SN3]{(\lam {\var}{\utm})\, \utmtwo \,\utm_1...\utm_n \in SN }
{\utm \sub{\var}{\utmtwo} \,\utm_1...\utm_n \in SN \quad \utmtwo \in SN}
\]
Let us consider rule ($SN1$). By induction there are $\tm_i$ such that $\judgRefine{\tm_i}{\utm_i}$ and
$\judgT{\tctx_i}{\tm_i}{\typ_i}$, so $\judgTs{\tctx_i}{\tms_i}{\typs_i}$. So, by Lemma \ref{lem:weak}, and rule ($\ruleTArre$),
$\judgT{\cris{(}\bigcup_{1 \leq i \leq n} \tctx_i\cris{)} \cup \var:\set{\typs_1 \arrow ...\arrow \typs_n \arrow \typtwo }}{\var^{\set{\typs_1 \arrow ...\arrow \typs_n \arrow \typtwo } }\tms_1 \,...\tms_n}{ \typtwo }$.
\noindent
Let us consider rule ($SN3$).
Let $\utm \sub{\var}{\utmtwo} \,\utm_1...\utm_n$ and $\utmtwo \in SN$.
By \ih, $\exists \tm, \tmtwo, \tctx, \tctx', \typ, \typtwo$ such that
$\judgRefine{\tm}{\utm \sub{\var}{\utmtwo} \,\utm_1...\utm_n}$, $\judgRefine{\tmtwo}{\utmtwo}$ and
$\judgT{\tctx}{\tm}{\typtwo}$,
$\judgT{\tctx'}{\tmtwo}{\typ}$. Note that $\judgT{\tctx'}{\tmtwo}{\typ}$ implies $\judgTs{\tctx'}{\tmstwo}{\typs}$,
where $\tmstwo=\set{\tmtwo}$ and $\typs = \set{\typ}$.
$\judgT{\tctx}{\tm}{\typtwo}$ and $\judgRefine{\tm}{\utm \sub{\var}{\utmtwo} \,\utm_1...\utm_n}$ imply $\tm= \tm'  \,\tm_1...\tm_n$, where $\judgRefine{\tm_i}{\utm_i}$ and $\tm'$ can be written as
$\tm'' \sub{\var^\typs}{\tmstwo}$.
By Lemma \ref{lem:weak}, $\judgT{\tctx \cup \tctx'}{\tm'' \sub{\var^\typs}{\tmstwo} \,\tm_1...\tm_n}{\typtwo}$
and $\judgTs{\tctx \cup \tctx'}{\tmstwo}{\typs}$, and the result follows by Lemma \ref{lem:inv-sub}.
The case of rule ($SN2$) follows directly by \ih.
\end{proof}

\section{Strong normalization via a decreasing measure}
\label{lambdaTm}
In this section, we prove strong normalization of the $\CalcLamT$ calculus
by providing an explicit decreasing measure,
adapting the ideas behind the $\meassym$ measure of~\cite{BarenbaumS23}
to the setting of idempotent intersection type systems.
The $\meassym$ measure is based on the fact that contracting a redex 
in the simply typed $\lambda$-calculus
cannot create a redex of higher or equal {\em degree},
where the degree of a redex is defined as the height of the type of its abstraction.
This observation was already known to Turing, as reported by Gandy~\cite{gandy80sn}.

This section is organized as follows.
In Section~\ref{subsec:CalcLamTm}, we enrich the syntax of $\CalcLamT$
by defining an auxiliary \emph{memory} calculus $\CalcLamTm$
which incorporates {\em wrappers} $\wrap{\tmstwo}$,
and we study some technical properties that are needed later.
In Section~\ref{subsec:simpd} we define an operation called \emph{full simplification}
for $\CalcLamTm$ terms,
which iteratively contracts in parallel all redexes of maximum degree,
showing that this yields the normal form of the term.
Finally, in Section~\ref{subsec:measure} we show that $\CalcLamT$ is SN
by defining a decreasing measure $\meassym$, which works by fully simplifying
a term in $\CalcLamTm$ and counting the number of remaining wrappers.

\subsection{The memory calculus $\CalcLamTm$}
\label{subsec:CalcLamTm}

\begin{defn}[The $\CalcLamTm$ calculus]
The sets of terms ($\tm,\tmtwo,\hdots$)
and set-terms ($\tms,\tmstwo,\hdots$)
are given by the following grammar:
\[
  \tm ::= \var^\typ
     \mid \lam{\var^\typs}{\tm}
     \mid \tm\,\tms
     \mid \tm\wrap{\tms}
     \qquad \qquad \qquad
  \tms ::= \set{\tm_j}_\jJ \qquad
\]
where $J$ stands for a finite set of indices.
We write $\sctx$ for a list of wrappers,
defined by the grammar $\sctx ::= \ctxhole \mid \sctx\wrap{\tms}$,
and $\tm\sctx$ for the term that results from extending $\tm$ with all the
wrappers in the list,
\ie $\tm(\ctxhole\wrap{\tmstwo_1}\hdots\wrap{\tmstwo_n})
   = \tm\wrap{\tmstwo_1}\hdots\wrap{\tmstwo_n}$.
Typing judgements are of the forms $\judgTm{\tctx}{\tm}{\typ}$
and $\judgTms{\tctx}{\tms}{\typs}$,
where types and typing contexts are defined as before.
The typing rules extend the $\CalcLamT$ type assignment system
of Definition~\ref{def_CalcLamT} with a typing rule $\ruleTmWrap$:
\[
  \infer[\ruleTmAx]{
    \judgTm{\tctx, \var : \typs}
    {\var^{\typtwo}}
    {\typtwo}
  }{
    \typtwo \in \typs
  }
  \qquad
  \infer[\ruleTmArri]{
    \judgTm{\tctx}{\lam{\var^\typs}{\tm}}{\typs\arrow\typtwo}
  }{
    \judgTm{\tctx,\var:\typs}{\tm}{\typtwo}
  }
\]
\[
  \infer[\ruleTmArre]{
    \judgTm{\tctx}{\app{\tm}{\tmstwo}}{\typtwo}
  }{
    \judgTm{\tctx}{\tm}{\typs\arrow\typtwo}
    \qquad
    \judgTms{\tctx}{\tmstwo}{\typs}
    \qquad
    \typs \neq \emptyset
  }
  \qquad
  \infer[\ruleTmWrap]{
    \judgTm{\tctx}{\tm\wrap{\tmstwo}}{\typ}
  }{
    \judgTm{\tctx}{\tm}{\typ}
    &
    \judgTms{\tctx}{\tmstwo}{\typstwo}
  }
\]
\[
  \infer[\ruleTmMulti]{
    \judgTms{\tctx}{\set{\tm_j}_\jJ}{\set{\typ_j}_\jJ}
  }{
    (\judgTm{\tctx}{\tm_j}{\typ_j})_\jJ
    \quad
    (\forall h,k \in J.\,\, h \neq k \implies \typ_h \neq \typ_k)
  }
\]
The operation of capture-avoiding substitution $\tm\sub{\var^\typs}{\tmstwo}$
is extended by declaring that
$(\tm\wrap{\tmsthree})\sub{\var^\typs}{\tmstwo} =
 \tm\sub{\var^\typs}{\tmstwo}\wrap{\tmsthree\sub{\var^\typs}{\tmstwo}}$.
Reduction in the $\CalcLamTm$-calculus, called $\symTm$-reduction,
is defined over terms and set-terms by the following rule,
closed by congruence under arbitrary contexts:
\[
  (\lam{\var^\typs}{\tm})\sctx\,\tmstwo
  \toTm
  \tm\sub{\var}{\tmstwo}\wrap{\tmstwo}\sctx
\]
Abstractions followed by lists of wrappers, \ie terms of the form
$(\lam{\var}{\tm})\sctx$ are called wrapped abstractions
or \defnn{w-abstractions} for short.
A \defnn{redex} is an expression matching the left-hand side of the
$\symTm$ rule, which must be an applied w-abstraction.
The \defnn{height} of a type $\typ$ (resp. set-type $\typs$)
is written $\height{\typ}$ (resp. $\height{\typs}$)
and defined as follows:
\[
  \begin{array}{rcl}
    \height{\basetyp}           & \eqdef & 0 \\
    \height{\typs\arrow\typtwo} & \eqdef & 1+\maxof{\height{\typs},\height{\typtwo}} \\
    \height{\set{\typ_1,\hdots,\typ_n}}
                             & \eqdef & \maxof{\height{\typ_1},\hdots,\height{\typ_n}} \\
  \end{array}
\]
We write $\typeof{\tm}$ for the type of $\tm$, which is uniquely
defined for typable terms.
The \defnn{degree} of a w-abstraction is the height of its type.
The degree of a redex is the degree of its w-abstraction.
The \defnn{max-degree} of a term $\tm \in \CalcLamTm$ is written $\maxdeg{\tm}$
and defined as the maximum degree of the redexes in $\tm$,
or $0$ if $\tm$ has no redexes.
This notion is also defined for set-terms ($\maxdeg{\tms}$)
and lists of wrappers ($\maxdeg{\sctx}$) in a similar way.
The \defnn{weight} of a term $\tm \in \CalcLamTm$ is written $\weight{\tm}$
and defined as the number of wrappers in $\tm$.
\end{defn}

\begin{exmp}
  In order to illustrate how the system works, and, especially, how wrappers are handled, we show two possible reductions for a term involving self-application and reductions inside the argument. We underline the contracted redex and highlight the new introduced memory at each step for clarity. Let
 $A  =  a \to a;\ A^A  =  A \to A;\ I^a  =  \lam{x^{\set{a}}}x;\ I^A  =  \lam{x^{\set{A}}}x;\ I^{A^A}  =  \lam{x^{\set{A^A}}}x$.
\end{exmp}

 \begin{center}
\begin{tikzpicture}[
  >=Latex,
  every node/.style={inner sep=2pt},
  step/.style={->, thick},
  node distance=5mm and 24mm
]
\node (M) {$ \underline{(\lam{x^{\set{A^A, A}}}{x^{A^A} \, x^A}) \set{I^{A^A}I^A,\dashuline{I^AI^a}}} $};

\node (M0) [below=of M] {};
\node (L0) [left=4cm of M0] {};
\node (L1) [below=0mm of L0] {$(I^{A^A}I^A)(I^AI^a) \bm{\wrap{\set{I^{A^A}I^A,\underline{I^AI^a}}}}$};
\node (L2) [below=of L1] {$(I^{A^A}I^A)(I^AI^a) \wrap{\set{\underline{I^{A^A}I^A},I^a\bm{\wrap{I^a}}}}$};
\node (L3) [below=of L2] {$(I^{A^A}I^A)(\underline{I^AI^a}) \wrap{\set{I^A\bm{\wrap{I^A}},I^a\wrap{I^a}}}$};
\node (L4) [below=of L3] {$(\underline{I^{A^A}I^A})(I^a\bm{\wrap{I^a}}) \wrap{\set{I^A\wrap{I^A},I^a\wrap{I^a}}}$};
\node (L5) [below=of L4] {};

\node (R0) [right=4cm of M0] {};
\node (R1) [below=of R0] {$ (\lam{x^{\set{A^A, A}}}{x^{A^A} \, x^A}) \set{\dashuline{I^{A^A}I^A},I^a\bm{\wrap{I^a}}} $};
\node (R2) [below=15mm of R1] {$ \dashuline{(\lam{x^{\set{A^A, A}}}{x^{A^A} \, x^A}) \set{I^A\bm{\wrap{I^A}},I^a\wrap{I^a}}} $};

\node (M4) [right=4cm of L5] {};
\node (M5) [below=0mm of M4] {$\underline{(I^A\bm{\wrap{I^A}})(I^a\wrap{I^a})} \bm{\wrap{\set{I^A\wrap{I^A},I^a\wrap{I^a}}}}$};
\node (M6) [below=of M5] {$I^a\wrap{I^a}\bm{\wrap{I^a\wrap{I^a}}}\wrap{I^A} \wrap{\set{I^A\wrap{I^A},I^a\wrap{I^a}}}$};

\draw[->] (M) -- (L1);
\draw[->,dashed] (M) -- (R1);

\draw[->] (L1) -- (L2);
\draw[->] (L2) -- (L3);
\draw[->] (L3) -- (L4);
\draw[->] (L4) -- (M5);
\draw[->] (M5) -- (M6);

\draw[->,dashed] (R1) -- (R2);
\draw[->,dashed] (R2) -- (M5);

\end{tikzpicture}
 \end{center}

\begin{exmp}
  The following reductions show how erased subterms are memorized by wrappers.
  \begin{enumerate}
    \item $\underline{(\lam{x^{\set A}}{y^{B}})((\lam{x^{\set B}}{z^{A}})w^B)}
      \quad \toTm \quad
      {y^B}\bm{\wrap{\underline{(\lam{x^{\set B}}{z^A})w^B}}}
      \quad \toTm \quad
      {y^B}\wrap{z^A\bm{\wrap{w^B}}}$
    \item $\underline{(\lam{x^{\set A}}{\lam{y^{\set B}}{x}})z^A}w^B
      \quad \toTm \quad \underline{(\lam{y^{\set B}}{z^A})\bm{\wrap{z^A}}w^B}
      \quad \toTm \quad {z^A}\bm{\wrap{w^B}}\wrap{z^A}$
  \end{enumerate}
\end{exmp}

\begin{rem}
Each step $\tm \toT \tmtwo$ in the $\CalcLamT$-calculus (without wrappers)
has a \defnn{corresponding} step $\tm \toTm \tmtwo'$
in the $\CalcLamTm$-calculus (with wrappers),
contracting the same redex but creating one wrapper.
For example, the step
$\app{(\lam{\var}{\varthree\,\var\,\var})}{I} \toT \varthree\,I\,I$
has a corresponding step
$\app{(\lam{\var}{\varthree\,\var\,\var})}{I} \toTm (\varthree\,I\,I)\wrap{I}$.
In particular an erasing reduction step in $\CalcLamT$ is mapped a non-erasing one 
in $\CalcLamTm$. For example, $\app{(\lam{x^\typs}{z^\typtwo})}{y^\typs} \toT z^\typtwo$,
while $\app{(\lam{x^\typs}{z^\typtwo})}{y^\typs} \toTm z^\typtwo \wrap{y^\typs}$.
\end{rem}

The following properties hold for the $\CalcLamTm$ calculus.
Subject expansion also holds for this calculus, but it is not needed
as part of the technical development.
See Appendices~\ref{appendix:section:CalcLamTm_subject_reduction} and~\ref{appendix:section:CalcLamTm_confluence} for detailed proofs.

\begin{restatable}[Subject reduction]{proposition}{calclamtmreduction}
  \label{CalcLamTm_subject_reduction}
  If $\tm \toTm \tm'$ and $\judgTm{\tctx}{\tm}{\typ}$,
  then $\judgTm{\tctx}{\tm'}{\typ}$.
\end{restatable}
\begin{restatable}[Confluence]{proposition}{lamtmconfluence}
  \label{LamTm:confluence}
  If $\tm_1 \toTm^* \tm_2$ and $\tm_1 \toTm^* \tm_3$,
  there exists a term $\tm_4$ such that
  $\tm_2 \toTm^* \tm_4$ and $\tm_3 \toTm^* \tm_4$.
\end{restatable}

\subsection{Simplification by complete developments}
\label{subsec:simpd}

We now introduce the operations of \emph{simplification} and \emph{full simplification} of a term. The \emph{simplification} computes the result of the complete development of the redexes of a given degree, \ie, reduces in parallel all the redexes of that degree. The \emph{full simplification} iteratively applies the simplification of a term, starting with the maximum degree and decreasing down to degree 1. Both operations are well-defined by recursion: simplification on the structure of terms,
and full simplification on the maximum degree of the term. Turing's observation (redex contractions cannot create redexes of higher or equal degree) is a key underlying property that ensures that full simplification works as intended, \ie, it computes the normal form of a term. The subsequent lemmas show it.

\begin{defn}[Simplification]
For each integer $d \geq 1$ we define an operation written $\simpd{-}$,
called \defnn{simplification of degree $d$},
that can be applied on terms ($\simpd{\tm}$),
set-terms ($\simpd{\tms}$), and lists of wrappers ($\simpd{\sctx}$),
mutually recursively as follows:
\[
  \begin{array}{rcl}
    \simpd{\var^\typ}
    & \eqdef &
    \var^\typ
  \\
    \simpd{\lam{\var^\typs}{\tm}}
    & \eqdef &
    \lam{\var^\typs}{\simpd{\tm}}
  \\
    \simpd{\app{\tm}{\tmstwo}}
    & \eqdef &
    \begin{cases}
      \simpd{\tm'}\sub{\var^\typs}{\simpd{\tmstwo}}\wrap{\simpd{\tmstwo}}\simpd{\sctx}
      & \text{if $\tm = (\lam{\var^\typs}{\tm'})\sctx$ of degree $d$}
    \\
      \app{\simpd{\tm}}{\simpd{\tmstwo}}
      & \text{otherwise}
    \end{cases}
  \\
    \simpd{\tm\wrap{\tmstwo}}
    & \eqdef &
    \simpd{\tm}\wrap{\simpd{\tmstwo}}
  \\
    \simpd{\set{\tm_1,\hdots,\tm_n}}
    & \eqdef &
    \set{\simpd{\tm_1},\hdots,\simpd{\tm_n}}
  \\
    \simpd{\ctxhole\wrap{\tms_1}\hdots\wrap{\tms_n}}
    & \eqdef &
    \ctxhole\wrap{\simpd{\tms_1}}\hdots\wrap{\simpd{\tms_n}}
  \end{array}
\]
If $\tm \in \CalcLamTm$ is a term and $D = \maxdeg{\tm}$,
the \defnn{full simplification of $\tm$}
is written $\simpfull{\tm}$
and defined as $\simpfull{\tm} \eqdef \simp{1}{\hdots\simp{D-1}{\simp{D}{\tm}}\hdots}$.
\end{defn}

\begin{lem}[Soundness of simplification]
\label{term_reduces_to_simpd}
If $\tm \in \CalcLamTm$
and $d \geq 1$
then $\tm \toTm^* \simpd{\tm}$.
\end{lem}
\begin{proof}
Straightforward by mutual induction on the term $\tm$, set-term $\tms$
and list of wrappers $\sctx$, generalizing the statement to set-terms and
lists of wrappers, \ie $\tms \toTm^* \simpd{\tms}$ and $\sctx \toTm^*
\simpd{\sctx}$.
\end{proof}

\begin{lem}[Creation of abstraction by substitution]
\label{substitution_creates_abstraction}
Let $\tm \in \CalcLamTm$
and $\tmstwo \subset \CalcLamTm$
such that $\tm$ is not a w-abstraction
but $\tm\sub{\var^\typs}{\tmstwo}$ is a w-abstraction.
Then $\tm = \var^\typtwo\sctx$ with $\typtwo\in\typs$.
\end{lem}
\begin{proof}
Let us write $\tm$
as of the form $\tm = \tm'\sctx$
where $\tm'$ is not a wrapper.
Note that $\tm'$ is not an abstraction, because
by hypothesis $\tm$ is not a w-abstraction.
Moreover, $\tm'$ cannot be an application, because
if $\tm' = \tm_1\,\tm_2$
then $\tm\sub{\var^\typs}{\tmstwo}
     = \tm'\sub{\var^\typs}{\tmstwo}(\sctx\sub{\var^\typs}{\tmstwo})
     = (\app{\tm_1\sub{\var^\typs}{\tmstwo}}{\tm_2\sub{\var^\typs}{\tmstwo}})
         (\sctx\sub{\var^\typs}{\tmstwo})$
would not be a w-abstraction.
The only remaining possibility is that $\tm'$ is a variable,
so $\tm' = \vartwo^\typtwo$.
If it is a variable other than $\var$, \ie $\vartwo \neq \var$,
then
$\tm\sub{\var^\typs}{\tmstwo}
 = \tm'\sub{\var^\typs}{\tmstwo}(\sctx\sub{\var^\typs}{\tmstwo})
 = \vartwo^\typtwo(\sctx\sub{\var^\typs}{\tmstwo})$
would not be a w-abstraction.
So we necessarily have that $\vartwo = \var^\typtwo$
with $\typtwo \in \typs$,
which concludes the proof.
\end{proof}

\begin{lem}[Bound for the max-degree of a substitution]
\label{max_degree_sub}
Let $\judgTms{\tctx}{\tmstwo}{\typs}$,
where $\maxdeg{\tmstwo} < d$ and $\height{\typs} < d$.
If $\judgTm{\tctx,\var:\typs}{\tm}{\typtwo}$
and $\maxdeg{\tm} < d$ then $\maxdeg{\tm\sub{\var^\typs}{\tmstwo}} < d$.
\end{lem}
\begin{proof}
We generalize the statement for set-terms, claiming that
$\judgTms{\tctx,\var:\typs}{\tms}{\typstwo}$ and $\maxdeg{\tms} < d$
imply $\maxdeg{\tms\sub{\var^\typs}{\tmstwo}} < d$.
The proof proceeds by simultaneous induction on the structure
of the term $\tm$ or set-term $\tms$.
The interesting cases are when the term is a variable or an application.
The remaining cases are straightforward by resorting to the \ih.
If $\tm$ is a \textbf{variable}, $\tm = \vartwo^\typthree$,
we consider two subcases, depending on whether $\vartwo = \var$ or not.
\begin{enumerate}
\item
  If $\var = \vartwo$,
  then $\typthree = \typtwo$
  and since $\judgTm{\tctx,\var:\typs}{\var}{\typtwo}$ holds,
  we have that $\typtwo \in \typs$.
  Hence there is a term $\tmtwo_0 \in \tmstwo$
  such that $\judgTm{\tctx}{\tmtwo_0}{\typtwo}$,
  and $\var^\typtwo\sub{\var^\typs}{\tmstwo} = \tmtwo_0$.
  In particular,
  by definition we have that $\maxdeg{\tmtwo_0} \leq \maxdeg{\tmstwo}$,
  and by hypothesis we have that $\maxdeg{\tmstwo} < d$,
  so
  $\maxdeg{\tm\sub{\var^\typs}{\tmstwo}}
  = \maxdeg{\var^\typtwo\sub{\var^\typs}{\tmstwo}}
  = \maxdeg{\tmtwo_0} < d$,
  as required.
\item
  If $\var \neq \vartwo$,
  then
  $\maxdeg{\tm\sub{\var^\typs}{\tmstwo}}
  = \maxdeg{\vartwo^\typthree\sub{\var^\typs}{\tmstwo}}
  = \maxdeg{\vartwo^\typthree} = \maxdeg{\tm} < d$.
\end{enumerate}
If $\tm$ is an \textbf{application}, $\tm = \app{\tmthree}{\tmsfour}$,
note that $\tmthree$ is a subterm of $\tm$
so $\maxdeg{\tmthree} \leq \maxdeg{\tm} < d$,
and similarly $\maxdeg{\tmsfour} \leq \maxdeg{\tm} < d$.
Thus applying the \ih on each subterm
we have that $\maxdeg{\tmthree\sub{\var^\typs}{\tmstwo}} < d$
and that $\maxdeg{\tmsfour\sub{\var^\typs}{\tmstwo}} < d$.
We consider three subcases, depending on whether
\itemNumber{1.} $\tmthree$ is a w-abstraction,
\itemNumber{2.} $\tmthree$ is not a w-abstraction and $\tmthree\sub{\var^\typs}{\tmstwo}$
                is a w-abstraction,
or
\itemNumber{3.} $\tmthree$ and $\tmthree\sub{\var^\typs}{\tmstwo}$
                are not w-abstractions:
\begin{enumerate}
\item
  If $\tmthree$ is a w-abstraction:
  then $\tmthree = (\lam{\vartwo^\typsthree}{\tmthree'})\sctx$
  is of type $\typsthree\arrow\typfour$.
  Let $k$ be the degree of the w-abstraction $\tmthree$,
  \ie the height of its type,
  $k = \height{\typsthree\arrow\typfour}$.
  By definition,
  $\maxdeg{\tm} = \maxdeg{\app{\tmthree}{\tmsfour}}
                = \maxof{k,\maxdeg{\tmthree},\maxdeg{\tmsfour}}$.
  Since $\maxdeg{\tm} < d$ by hypothesis, in particular we have that $k < d$.
  Moreover, note that
  $\tmthree\sub{\var^\typs}{\tmstwo}
  = (\lam{\vartwo^\typsthree}{\tmthree'\sub{\var^\typs}{\tmstwo}})(\sctx\sub{\var^\typs}{\tmstwo})$,
  so $\tmthree\sub{\var^\typs}{\tmstwo}$
  is a w-abstraction, and it is of degree $k$
  because substitution preserves the type of a term.
  Therefore:
  \[
    \begin{array}{rclcl}
      \maxdeg{\tm\sub{\var^\typs}{\tmstwo}}
    & = &
      \maxdeg{\app{\tmthree\sub{\var^\typs}{\tmstwo}}{\tmsfour\sub{\var^\typs}{\tmstwo}}}
    \\
    & = &
      \maxof{k,\maxdeg{\tmthree\sub{\var^\typs}{\tmstwo}},\maxdeg{\tmsfour\sub{\var^\typs}{\tmstwo}}}
    & < &
      d
    \end{array}
  \]
  The last step is justified because we have already noted that
      $k < d$
  and $\maxdeg{\tmthree\sub{\var^\typs}{\tmstwo}} < d$
  and $\maxdeg{\tmsfour\sub{\var^\typs}{\tmstwo}} < d$.
\item
  If $\tmthree$ is not a w-abstraction
  and $\tmthree\sub{\var^\typs}{\tmstwo}$ is a w-abstraction:
  then by Lemma~\ref{substitution_creates_abstraction}
  $\tmthree$ must be of the form $\var^\typthree\sctx$ with $\typ\in\typthree$.
  Let $k$ be the degree of the w-abstraction $\tmthree\sub{\var^\typs}{\tmstwo}$.
  Then $k$ is the height of the type of $\tmthree\sub{\var^\typs}{\tmstwo}$,
  that is, $k = \height{\typthree} \leq \height{\typs} < d$.
  To conclude, note that:
  \[
    \begin{array}{rclcl}
      \maxdeg{\tm\sub{\var^\typs}{\tmstwo}}
    & = &
      \maxdeg{\app{\tmthree\sub{\var^\typs}{\tmstwo}}{\tmsfour\sub{\var^\typs}{\tmstwo}}}
    \\
    & = &
      \maxof{
        k,
        \maxdeg{\tmthree\sub{\var^\typs}{\tmstwo}},
        \maxdeg{\tmsfour\sub{\var^\typs}{\tmstwo}}
      }
    & < &
      d
    \end{array}
  \]
  The last step is justified because we have already noted that
      $k < d$
  and $\maxdeg{\tmthree\sub{\var^\typs}{\tmstwo}} < d$
  and $\maxdeg{\tmsfour\sub{\var^\typs}{\tmstwo}} < d$.
\item
  If $\tmthree$ and $\tmthree\sub{\var^\typs}{\tmstwo}$ are not a w-abstractions:
  then
  \[
    \begin{array}{rclcl}
      \maxdeg{\tm\sub{\var^\typs}{\tmstwo}}
    & = &
      \maxdeg{\app{\tmthree\sub{\var^\typs}{\tmstwo}}{\tmsfour\sub{\var^\typs}{\tmstwo}}}
    \\
    & = &
      \maxof{\maxdeg{\app{\tmthree\sub{\var^\typs}{\tmstwo}}},
             \maxdeg{\tmsfour\sub{\var^\typs}{\tmstwo}}}
    & < &
      d
    \end{array}
  \]
  The last step is justified because we have already noted that
      $\maxdeg{\tmthree\sub{\var^\typs}{\tmstwo}} < d$
  and $\maxdeg{\tmsfour\sub{\var^\typs}{\tmstwo}} < d$.
  \qedhere
\end{enumerate}
\end{proof}

\begin{lem}[Simplification does not create abstractions]
\label{simplification_does_not_create_abstractions}
Suppose that $\tm \in \CalcLamTm$ is not a w-abstraction
and $\maxdeg{\tm} \leq d$
and $\height{\typeof{\tm}} \geq d$.
Then $\simpd{\tm}$ is not a w-abstraction.
\end{lem}
\begin{proof}
We proceed by induction on $\tm$.
The interesting case is when $\tm$ is a redex of degree $d$.
We claim that this case is impossible.
Indeed, suppose that $\tm = (\lam{\var^\typs}{\tmtwo})\sctx\,\tmsthree$,
where $\lam{\var^\typs}{\tmtwo}$ is of type $\typs\arrow\typtwo$
so that $\tm$ is of type $\typtwo$.
The degree of the redex is $d = \height{\typs\arrow\typtwo}$
and $\height{\typeof{\tm}} = \height{\typtwo} < d$.
However, by hypothesis, $\height{\typeof{\tm}} \geq d$,
from which we obtain a contradiction.
The remaining cases are straightforward resorting to the \ih.
\end{proof}

\begin{lem}[Simplification decreases the max-degree]
\label{simp_decreases_maxdeg}
If $d \geq 1$ and $\maxdeg{\tm} \leq d$ then $\maxdeg{\simpd{\tm}} < d$.
\end{lem}
\begin{proof}
  We proceed by simultaneous induction,
generalizing the statement for set-terms
($\maxdeg{\tms} \leq d$ implies $\maxdeg{\simpd{\tms}} < d$)
and lists of wrappers
($\maxdeg{\sctx} \leq d$ implies $\maxdeg{\simpd{\sctx}} < d$).
The interesting cases are when $\tm$ is a variable or an application.
The remaining cases are straightforward by \ih.
If $\tm$ is a \textbf{variable}, $\tm = \var^\typ$,
then
  $
    \maxdeg{\simpd{\var^\typ}}
  = \maxdeg{\var^\typ}
  = 0
  < d
  $
  because $d \geq 1$ by hypothesis.
If $\tm$ is an \textbf{application}, $\tm = \app{\tmtwo}{\tmsthree}$,
we consider two cases, depending on whether $\tmtwo$ is a w-abstraction of degree $d$
or not:
  \begin{enumerate}
  \item
    If $\tmtwo = (\lam{\var^\typs}{\tmtwo'})\sctx$ is a w-abstraction of degree $d$:
    then note that $\tmtwo'$ is a subterm of $\tm$
    which implies that $\maxdeg{\tmtwo'} \leq \maxdeg{\tm} \leq d$,
    and this in turn implies by \ih that $\maxdeg{\simpd{\tmtwo'}} < d$.
    Similarly,
    we can note that $\maxdeg{\sctx} \leq d$ so by \ih $\maxdeg{\simpd{\sctx}} < d$,
    and that $\maxdeg{\tmsthree} \leq d$ so by \ih $\maxdeg{\simpd{\tmsthree}} < d$.
    Since the term is well-typed,
    the type of the w-abstraction $\tmtwo = (\lam{\var^\typs}{\tmtwo'})\sctx$
    is of the form $\typs\arrow\typtwo$,
    while the type of $\tmsthree$ is $\typs$.
    Moreover, $\height{\typs\arrow\typtwo} = d$ by hypothesis,
    so in particular $\height{\typs} < d$.
    From this we obtain by Lemma~\ref{max_degree_sub} that
    $\maxdeg{\simpd{\tmtwo'}\sub{\var^\typs}{\simpd{\tmsthree}}} < d$.
    Finally, we have:
    \[
      \begin{array}{rclcl}
        \maxdeg{\simpd{\app{\tmtwo}{\tmsthree}}}
      & = &
        \maxdeg{
          \simpd{\tmtwo'}\sub{\var^\typs}{\simpd{\tmsthree}}
                         \wrap{\simpd{\tmsthree}}
                         \simpd{\sctx}
        }
      \\
      & = &
        \maxof{
          \maxdeg{\simpd{\tmtwo'}\sub{\var^\typs}{\simpd{\tmsthree}}},
          \maxdeg{\simpd{\tmsthree}},
          \maxdeg{\simpd{\sctx}}
        }
        <
        d
      \end{array}
    \]
  \item
    If $\tmtwo$ is not a w-abstraction of degree $d$:
    since $\tmtwo$ is a subterm of $\tm$
    we have that $\maxdeg{\tmtwo} \leq \maxdeg{\tm} \leq d$
    so by \ih $\maxdeg{\simpd{\tmtwo}} < d$.
    Similarly, $\maxdeg{\tmsthree} \leq \maxdeg{\tm} \leq d$
    so by \ih $\maxdeg{\simpd{\tmsthree}} < d$.
    Let the type of $\tmtwo$ be of the form $\typs\arrow\typtwo$
    and let $k := \height{\typs\arrow\typtwo}$.
    We consider two cases, depending on whether $k < d$ or $k \geq d$.

      If $k = \height{\typs\arrow\typtwo} < d$:
      then by Lemma~\ref{term_reduces_to_simpd}
      we have $\tmtwo \toTm^* \simpd{\tmtwo}$
      and by Subject~Reduction~(\ref{CalcLamTm_subject_reduction})
      the type of $\simpd{\tmtwo}$ is also $\typs\arrow\typtwo$.
      Hence:
      \[
          \maxdeg{\simpd{\tm}}
        =
          \maxdeg{\simpd{\app{\tmtwo}{\tmsthree}}}
        =
          \maxdeg{\app{\simpd{\tmtwo}}{\simpd{\tmsthree}}}
        =
          \maxof{k,\maxdeg{\tmtwo},\maxdeg{\tmsthree}}
        <
          d
      \]

      On the other hand, if $k = \height{\typs\arrow\typtwo} \geq d$,
      we claim that $\tmtwo$ cannot be a w-abstraction.
      First, note by hypothesis that $\tmtwo$ is not a w-abstraction
      of degree $k = d$.
      Second, $\tmtwo$ cannot be a w-abstraction of degree $k > d$,
      as we would have that $\app{\tmtwo}{\tmsthree}$
      is a redex of degree $k$,
      so $\maxdeg{\app{\tmtwo}{\tmsthree}} \geq k$.
      This would imply that
      $d < k \leq \maxdeg{\app{\tmtwo}{\tmsthree}} \leq d$,
      a contradiction.
      Since $\tmtwo$ is not a w-abstraction,
      its type is $\typs\arrow\typtwo$ of height $k \geq d$,
      and $\maxdeg{\tmtwo} \leq d$,
      by Lemma~\ref{simplification_does_not_create_abstractions}
      we have that $\simpd{\tmtwo}$ is not a w-abstraction.
      Hence $\app{\simpd{\tmtwo}}{\simpd{\tmsthree}}$ is not a redex,
      and
      $
          \maxdeg{\simpd{\tm}}
        =
          \maxdeg{\simpd{\app{\tmtwo}{\tmsthree}}}
        =
          \maxdeg{\app{\simpd{\tmtwo}}{\simpd{\tmsthree}}}
        =
          \maxof{\maxdeg{\simpd{\tmtwo}},\maxdeg{\simpd{\tmsthree}}}
        <
          d
      $. \qedhere
  \end{enumerate}
\end{proof}

\begin{prop}[Full simplification yields the normal form]
\label{full_simplification_yields_nf}
Let $\tm \in \CalcLamTm$.
Then $\tm \toTm^* \simpfull{\tm}$
and $\simpfull{\tm}$ is in $\toTm$-normal form.
\end{prop}
\begin{proof}
Let $d$ be the max-degree of $\tm$
and define
$\simp{>k}{\tm} \eqdef \simp{k+1}{\hdots\simp{d-1}{\simp{d}{\tm}}\hdots}$
for each $k$ such that $0 \leq k \leq d$.
Note that $\simp{>d}{\tm} = \tm$ and $\simp{>0}{\tm} = \simpfull{\tm}$.
To show that $\tm \toTm^* \simpfull{\tm}$, note that 
$\simp{>k}{\tm} \toTm^* \simp{k}{\simp{>k}{\tm}} = \simp{>k-1}{\tm}$ for each $1
\leq k \leq d$ by Lemma~\ref{term_reduces_to_simpd},
so:
  \[
    \tm = \simp{>d}{\tm}
    \toTm^* \simp{>d-1}{\tm}
    \toTm^* \hdots
    \toTm^* \simp{>k}{\tm}
    \toTm^* \hdots
    \toTm^* \simp{>0}{\tm}
    = \simpfull{\tm}
  \]
To show that $\simpfull{\tm}$ is a $\toTm$-normal form,
we claim that for each $0 \leq k \leq d$
we have that $\maxdeg{\simp{>k}{\tm}} \leq k$.
We proceed by induction on $d - k$.
If $d - k = 0$, we have that $k = d$,
so $\maxdeg{\simp{>d}{\tm}} = \maxdeg{\tm} = d$,
since $d$ is the max-degree of $\tm$.
If $d - k > 0$, we have that $0 \leq k < d$.
By \ih, $\maxdeg{\simp{>k+1}{\tm}} \leq k+1$.
Then $\maxdeg{\simp{>k}{\tm}}
= \maxdeg{\simp{k+1}{\simp{>k+1}{\tm}}}
< k+1$
by Lemma~\ref{simp_decreases_maxdeg}.
This means that $\maxdeg{\simp{>k}{\tm}} \leq k$, as required.
\qedhere
\end{proof}

\subsection{The decreasing measure}
\label{subsec:measure}

The $\meassym$-measure for $\CalcLamT$ is defined in two steps: computing the normal form in $\CalcLamTm$, and counting the number of wrappers in it. The first step was taken care of in the previous section.
The second one is simple, it just requires to analyze the resulting term of the full simplification.

The core of this section is proving that the measure decreases, \ie, $t \toT s$ implies $\meas t > \meas s$. The proof is carried out by observing how $\toT$ and $\toTm$ differ, and how wrappers in $\CalcLamTm$ are handled.
Confluence holds in $\CalcLamTm$, so terms share normal form, and therefore full simplification, with all their reducts.
The intuition behind the proof is that, if we have $t \toT s$, then the corresponding step $t \toTm s'$ is responsible for the presence of at least one wrapper in $\simpfull{t}$ that is absent from $\simpfull{s}$.

\begin{defn}[Forgetful reduction]
We define a binary relation $\tm \tof \tmtwo$
between typable terms as the closure by congruence
under arbitrary contexts of the axiom
$\tm\wrap{\tmstwo} \tof \tm$.
\end{defn}

\begin{rem} 
\label{forgetful_reduction_does_not_create_redexes}
If $\tm$ is in $\toTm$-normal form
and $\tm \tof \tmtwo$
then $\tmtwo$ is also in $\toTm$-normal form.
\end{rem}

\begin{rem} 
\label{forgetful_reduction_decreases_weight}
If $\tm \tof \tmtwo$
then $\weight{\tm} > \weight{\tmtwo}$.
\end{rem}

\begin{lem}[Reduce/forget lemma]
\label{reduce_forget_lemma}
Let $\tm \toT \tmtwo$ be a step in the $\CalcLamT$-calculus (without wrappers)
and consider its corresponding step $\tm \toTm \tmtwo'$.
Then $\tmtwo' \tof \tmtwo$.
\end{lem}
\begin{proof}
Straightforward by induction on $\tm$.
\end{proof}

\begin{lem}[Forgetful reduction commutes with reduction]
\label{forgetful_reduction_commutes_with_reduction}
If $\tm_1 \toTm^* \tm_2$ and $\tm_1 \tof^+ \tm_3$
there exists $\tm_4$ such that
$\tm_2 \tof^+ \tm_4$ and $\tm_3 \toTm^* \tm_4$.
\end{lem}
\begin{proof}
It suffices to prove a local commutation result, namely that
if $\tm_1 \toTm \tm_2$ and $\tm_1 \tof \tm_3$
there exists $\tm_4$ such that
$\tm_2 \tof^+ \tm_4$ and $\tm_3 \toTm^= \tm_4$.
Graphically:
\[
  \diagramCommuteTmFTmFp{\tm_1}{\tm_2}{\tm_3}{\tm_4}
\]
We proceed by induction on $\tm_1$.
\begin{enumerate}
\item $\tm_1 = \var^\typs$: this case is impossible, since $\var$ does not reduce.
\item $\tm_1 = \lam{\var^\typs}{\tm_1'}$:
  then $\tm_2 = \lam{\var^\typs}{\tm_2'}$ and $\tm_3 = \lam{\var^\typs}{\tm_3'}$,
  with $\tm_1' \toTm \tm_2'$ and $\tm_1' \tof \tm_3'$. By \ih there exists
  $\tm_4'$ s.t. $\tm_3' \toTm^= \tm_4'$ and $\tm_2' \tof^+ \tm_4'$.
  By congruence, the following holds
  \[
    \diagramCommuteTmFTmFp
      {\lam{\var^\typs}{\tm_1'}}
      {\lam{\var^\typs}{\tm_2'}}
      {\lam{\var^\typs}{\tm_3'}}
      {\lam{\var^\typs}{\tm_4'}}
  \]
\item $\tm_1 = \app{\tm_{11}}{\tms_{12}}$. We consider two subcases here,
  depending on $\tm_{11}$ being a w-abstraction and on which redex is contracted.
  \begin{enumerate}
\item $\tm_1 = \app{(\lam{\var^\typs}{\tm_{11}})\sctx}\tms_{12}$ and the
  contracted redex is at head position. Then
    \[
    \xymatrix@C=.2em@R=.2em{
      \app{(\lam{\var^\typs}{\tm_{11}})\sctx}\tms_{12} &
      \toTm &
      \tm_{11}\sub{\var^\typs}{\tms_{12}}\wrap{\tms_{12}}\sctx \\
      \tofvert  & & \\
      \app{(\lam{\var^\typs}{\tm_{11}'})\sctx'}\tms_{12}' &  &  \\
    } \]
    where since there is a single $\tof$-step, only one of $\tm_{11}'$,
    $\sctx'$, or $\tms_{12}'$, can be different from $\tm_{11}$, $\sctx$, or
    $\tms_{12}$, respectively. If the $\tof$-step occurs in $\tm_{11}$ or
    $\sctx$, then the same step can be reproduced from $\tm_2$. Else,
    considering $\tms_{12}$ can appear more than once in $\tm_2$, the
    $\tof$-step may have to be reproduced several times, hence the $\tof^+$.
    Then we have that
    \[
      \diagramCommuteTmFTmOneFp
      {\app{(\lam{\var^\typs}{\tm_{11}})\sctx}\tms_{12}}
      {\tm_{11}\sub{\var^\typs}{\tms_{12}}\wrap{\tms_{12}}\sctx}
      {\app{(\lam{\var^\typs}{\tm_{11}'})\sctx'}\tms_{12}'}
      {\tm_{11}'\sub{\var^\typs}{\tms_{12}'}\wrap{\tms_{12}'}\sctx'}
    \]
  \item $\tm_{11}$ is not a w-abstraction or the contracted redex is not the one
    at head position. We consider two subcases.
    \begin{enumerate}
    \item Both steps occur in the same subterm. It suffices to resort to the \ih:\\
      \begin{minipage}{.45\textwidth}
        \begin{displaymath}
          \diagramCommuteTmFTmFp
          {\tm_{11}\tms_{12}}
          {\tm_{11}'\tms_{12}}
          {\tm_{11}''\tms_{12}}
          {\tm_{11}'''\tms_{12}}
        \end{displaymath}
      \end{minipage}
      \begin{minipage}{.45\textwidth}
        \begin{displaymath}
          \diagramCommuteTmFTmFp
          {\tm_{11}\tms_{12}}
          {\tm_{11}\tms_{12}'}
          {\tm_{11}\tms_{12}''}
          {\tm_{11}\tms_{12}'''}
        \end{displaymath}
      \end{minipage}
  \item
    The steps occur in different subterms.
    Then the diagram can be closed immediately:
    \begin{minipage}{.45\textwidth}
      \begin{displaymath}
        \diagramCommuteTmFTmOneFOne
        {\tm_{11}\tms_{12}}
        {\tm_{11}'\tms_{12}}
        {\tm_{11}\tms_{12}'}
        {\tm_{11}'\tms_{12}'}
      \end{displaymath}
    \end{minipage}
    \begin{minipage}{.45\textwidth}
    \begin{displaymath}
      \diagramCommuteTmFTmOneFOne
        {\tm_{11}\tms_{12}}
        {\tm_{11}\tms_{12}'}
        {\tm_{11}'\tms_{12}}
        {\tm_{11}'\tms_{12}'}
      \end{displaymath}
    \end{minipage}\\
  \end{enumerate}
\end{enumerate}
\item $\tm_1 = \tm_{11}\wrap{\tms_{12}}$. We consider three general subcases:
  \begin{enumerate}
  \item Both steps occur in the same subterm. It suffices to resort to the \ih: \\
    \begin{minipage}{.45\textwidth}
      \begin{displaymath}
        \diagramCommuteTmFTmFp
        {\tm_{11}\wrap{\tms_{12}}}
        {\tm_{11}'\wrap{\tms_{12}}}
        {\tm_{11}''\wrap{\tms_{12}}}
        {\tm_{11}'''\wrap{\tms_{12}}}
      \end{displaymath}
    \end{minipage}
    \begin{minipage}{.45\textwidth}
      \begin{displaymath}
        \diagramCommuteTmFTmFp
        {\tm_{11}\wrap{\tms_{12}}}
        {\tm_{11}\wrap{\tms_{12}'}}
        {\tm_{11}\wrap{\tms_{12}''}}
        {\tm_{11}\wrap{\tms_{12}'''}}
      \end{displaymath}
    \end{minipage}
  \item The steps occur in different subterms. Then the diagram can be closed immediately:
    \begin{minipage}{.45\textwidth}
      \begin{displaymath}
        \diagramCommuteTmFTmOneFOne
        {\tm_{11}\wrap{\tms_{12}}}
        {\tm_{11}'\wrap{\tms_{12}}}
        {\tm_{11}\wrap{\tms_{12}'}}
        {\tm_{11}'\wrap{\tms_{12}'}}
      \end{displaymath}
    \end{minipage}
    \begin{minipage}{.45\textwidth}
      \begin{displaymath}
        \diagramCommuteTmFTmOneFOne
        {\tm_{11}\wrap{\tms_{12}}}
        {\tm_{11}\wrap{\tms_{12}'}}
        {\tm_{11}'\wrap{\tms_{12}}}
        {\tm_{11}'\wrap{\tms_{12}'}}
      \end{displaymath}
    \end{minipage}\\
  \item The $\tof$-step occur at the root. We perform the $\toTm$ step, or not.\\
    \begin{minipage}{.45\textwidth}
      \begin{displaymath}
        \diagramCommuteTmFTmOneFOne
        {\tm_{11}\wrap{\tms_{12}}}
        {\tm_{11}'\wrap{\tms_{12}}}
        {\tm_{11}}
        {\tm_{11}'}
      \end{displaymath}
    \end{minipage}
    \begin{minipage}{.45\textwidth}
      \begin{displaymath}
        \diagramCommuteTmFEqFOne
        {\tm_{11}\wrap{\tms_{12}}}
        {\tm_{11}\wrap{\tms_{12}'}}
        {\tm_{11}}
        {\tm_{11}}
      \end{displaymath}
    \end{minipage}\\
  \end{enumerate}
\item $\tm_1 = \set{\tmtwo_1,\hdots,\tmtwo_n}$. We consider two subcases.
  \begin{enumerate}
  \item Both steps occur in the same subterm. It suffices to resort to the \ih:
  \[
    \diagramCommuteTmFTmFp
    {\set{\tmtwo_1,\hdots,\tmtwo_i,\hdots,\tmtwo_n}}
    {\set{\tmtwo_1,\hdots,\tmtwo_i',\hdots,\tmtwo_n}}
    {\set{\tmtwo_1,\hdots,\tmtwo_i'',\hdots,\tmtwo_n}}
    {\set{\tmtwo_1,\hdots,\tmtwo_i''',\hdots,\tmtwo_n}}
  \]
\item The steps occur in different subterms. Then the diagram can be closed immediately.
      Indeed, let $i,j \in 1..n$:
  \[
    \diagramCommuteTmFTmOneFOne
    {\set{\tmtwo_1,\hdots,\tmtwo_i,\hdots,\tmtwo_j,\hdots,\tmtwo_n}}
    {\set{\tmtwo_1,\hdots,\tmtwo_i',\hdots,\tmtwo_j,\hdots,\tmtwo_n}}
    {\set{\tmtwo_1,\hdots,\tmtwo_i,\hdots,\tmtwo_j',\hdots,\tmtwo_n}}
    {\set{\tmtwo_1,\hdots,\tmtwo_i',\hdots,\tmtwo_j',\hdots,\tmtwo_n}}
  \]
  \end{enumerate}
\end{enumerate}
\end{proof}

\begin{lem}[Properties of the full simplification]
\label{simpfull_properties}
\quad\\
\itemNumber{1.} If $\tm \toTm \tmtwo$ then $\simpfull{\tm} = \simpfull{\tmtwo}$.
\itemNumber{2.} If $\tm \tof \tmtwo$ then $\simpfull{\tm} \tof^+ \simpfull{\tmtwo}$.
\end{lem}
\begin{proof}
We prove each item separately.
  First,
  suppose that $\tm \toTm \tmtwo$.
  By the fact that full simplification yields the normal form
  of a term (Proposition~\ref{full_simplification_yields_nf}),
  we know that $\tm \toTm^* \simpfull{\tm}$
  and $\tmtwo \toTm^* \simpfull{\tmtwo}$,
  where both $\simpfull{\tm}$ and $\simpfull{\tmtwo}$ are
  in $\toTm$-normal form.
  By Confluence (Proposition~\ref{LamTm:confluence})
  we conclude that $\simpfull{\tm} = \simpfull{\tmtwo}$.

  Second, suppose that $\tm \tof \tmtwo$.
  By the fact that full simplification yields the normal form
  of a term~(Proposition~\ref{full_simplification_yields_nf}),
  we know that $\tm \toTm^* \simpfull{\tm}$,
  where $\simpfull{\tm}$ is in $\toTm$-normal form.
  Since forgetful reduction commutes
  with reduction (Lemma~\ref{forgetful_reduction_commutes_with_reduction}),
  we have that there exists a term $\tmthree \in \CalcLamTm$
  such that $\tmtwo \toTm^* \tmthree$ and $\simpfull{\tm} \tof^+ \tmthree$.
  By Remark~\ref{forgetful_reduction_does_not_create_redexes},
  since $\simpfull{\tm}$ is in $\toTm$-normal form,
  $\tmthree$ is also in $\toTm$-normal form.
  Moreover, $\tmtwo \toTm^* \simpfull{\tmtwo}$
  where $\simpfull{\tmtwo}$ is in $\toTm$-normal form by
  Proposition~\ref{full_simplification_yields_nf}.
  So by Confluence (Proposition~\ref{LamTm:confluence})
  we have that $\simpfull{\tmtwo} = \tmthree$.
  Then we have that
  $\simpfull{\tm} \tof^+ \tmthree = \simpfull{\tmtwo}$,
  as required. \qedhere
\end{proof}

\begin{defn}[The $\meassym$-measure]
  The measure of a term $\tm \in \CalcLamT$
  is defined as the weight of its full simplification,
  where we recall that the weight is the number of wrappers
  in a term:
  \[
    \meas{\tm} \eqdef \weight{\simpfull{\tm}}
  \]
\end{defn}

\begin{thm}[The $\meassym$-measure is decreasing]
Let $\tm,\tmtwo \in \CalcLamT$ be typable terms
of the $\CalcLamT$-calculus (without wrappers).
If $\tm \toT \tmtwo$ then $\meas{\tm} > \meas{\tmtwo}$.
\end{thm}
\begin{proof}
Let $\tm \toT \tmtwo$
and consider the corresponding step $\tm \toTm \tmtwo'$.
By the reduce/forget lemma (\ref{reduce_forget_lemma})
we have that $\tmtwo' \tof \tmtwo$.
By Lemma~\ref{simpfull_properties} (\lemmaPart{1})
we have that $\simpfull{\tm} = \simpfull{\tmtwo'}$.
By Lemma~\ref{simpfull_properties} (\lemmaPart{2})
we have that $\simpfull{\tmtwo'} \tof^+ \simpfull{\tmtwo}$.
Hence $\simpfull{\tm} \tof^+ \simpfull{\tmtwo}$.
By Remark~\ref{forgetful_reduction_decreases_weight}
we conclude that
  $
    \meas{\tm}
  = \weight{\simpfull{\tm}}
  > \weight{\simpfull{\tmtwo}}
  = \meas{\tmtwo}
  $.
\end{proof}

\begin{cor}\label{cor:SN}
$\CalcLamT$ is strongly normalizing.
\end{cor}

\section{Conclusion}
\label{concl}
In this section we provide a proof of the fact that $\CalcLame$ characterizes strong normalization, based on the correspondence between it and $\CalcLamT$. This proof justifies the design of $\CalcLame$, showing that it is a variant of the original system in \cite{DBLP:journals/ndjfl/CoppoD80}, w.r.t. the typability power.
Note that we make use of $\CalcLamT$ to prove the second item for uniformity purposes; it can also be proved entirely within $\CalcLame$.
\begin{cor}
$\CalcLame$ characterizes strong normalization.
\end{cor}
\begin{proof}
\begin{itemize}
\item Let $\utm \in \CalcLam$ be a $\lambda$-term typable in the $\CalcLame$
system, \ie such that $\judge{\tctx}{\utm}{\typ}$.
By the Correspondence Theorem~(\ref{lami_lamT_correspondence}.\ref{corr-one}),
there is a term $\tm \in \CalcLamT$
typable in the $\CalcLamT$ system such that
$\judgT{\tctx}{\tm}{\typ}$ and $\judgRefine{\tm}{\utm}$.
By Corollary~\ref{cor:SN}, $\tm$ is strongly normalizing.
From the Simulation Theorem~(\ref{lami_lamT_simulation}), we can observe that
an infinite reduction sequence
$\utm \tobeta \utm_1 \tobeta \utm_2 \tobeta \hdots$
would induce an infinite reduction sequence
$\tm \toT^+ \tm_1 \toT^+ \tm_2 \toT^+ \hdots$, and this would contradict strong normalization of $\CalcLamT$.
Hence $\utm$ is strongly normalizing.
\item Let $\utm\in \CalcLam$ be strongly normalizing. By Lemma \ref{lem:strong}, there is 
$\tm \in \CalcLamT$ such that
  $\judgRefine{\tm}{\utm}$ and $\judgT{\tctx}{\tm}{\typ}$,
  for some $\tctx$ and $\typ$. Then $\tm$ is uniform, and, by Theorem \ref{lami_lamT_correspondence}.\ref{corr-two},
 $\judgT{\tctx}{\tm}{\typ}$ implies $\judge{\tctx}{\typerase{\tm}}{\typ}$. 
 Since $\typerase{\tm} =\utm$, the proof is given.
\end{itemize}
\end{proof}

As we said in Section~\ref{intro}, this is not a completely new result.
The novelty of our work comes from the technique we propose and the simplicity of its output.
The strong normalization property for idempotent intersection type systems has already been
proved in various papers using both semantical and syntactical approaches, and the fact that $\CalcLame$
is a variant of them can be considered folklore. 
The semantical approach relies on arguments like computability or reducibility
candidates (\eg~\cite{Pottinger80,barendregt2013lambda}), while the syntactical
ones are based on a measure that decreases with the $\beta$-reduction; as far as we
know, there are three other syntactical proofs
\cite{DBLP:conf/lics/KfouryW95,DBLP:conf/tlca/Boudol03,DBLP:journals/mscs/BucciarelliPS03}.
Here we supply a further syntactical proof, which uses as key ingredient a
Church version of idempotent intersection types that has good proof-theoretical
properties. The notion of measure obtained is simpler
than~\cite{DBLP:conf/lics/KfouryW95,DBLP:conf/tlca/Boudol03}, being a natural
number, and operates directly within the intersection system, allowing for
refinement where measures based on the translation
in~\cite{DBLP:journals/mscs/BucciarelliPS03} are constrained.

In the future, we would like to explore the possibility of providing an exact
decreasing measure, \ie a measure whose result is an exact upper bound for the
length of the longest reduction chain starting from the term. While this
refinement can be easily done for simple types, through a modification of the
method in \cite{BarenbaumS23}, the extension to intersection types presents
challenges. In this case, the number of reduction steps in the derivation tree
in general is bigger than the number of $\beta$-reductions in the term, so it would
be necessary to change the operational behaviour of $\CalcLamT$.

Moreover, it would be interesting to study if the techniques given in this paper
can be adapted to intersection type systems that characterize solvable terms,
such as that of~\cite{CoppoDV81}, for which no syntactic proofs of normalization
have been proposed, as far as we know.

\bibliographystyle{./entics}
\bibliography{biblio}

\appendix

\section{Technical appendix}

\subsection{Proofs of Section~\ref{lambdaT} --- An intrinsically typed presentation of idempotent intersection types}
\label{appendix:lambdaT}

In this section we give detailed proofs of the results
about the $\CalcLamT$-calculus stated in Section~\ref{lambdaT}.

\subsubsection{Subject reduction}
\label{appendix:section:CalcLamT_subject_reduction}

\subjectreduction*

\begin{proof}
We generalize the statement, proving also that
if $\tms \toT \tms'$ and $\judgTs{\tctx}{\tms}{\typs}$
then $\judgTs{\tctx}{\tms'}{\typs}$.
The proof proceeds by induction on the definition of  $\toT$.
\begin{enumerate}
\item If $\tm$ is a redex, i.e, $\tm=\app{(\lam{\var^\typstwo}{\tmtwo})}{\tmsthree}$, then the derivation proving 
$\judgT{\tctx}{\tm}{\typ}$ is:
\[
\infer{\judgT{\tctx}{\app{(\lam{\var^\typstwo}{\tmtwo})}{\tmsthree}}{\typ}}
{\infer{\judgT{\tctx}{\lam{\var^\typstwo}{\tmtwo}}{\typstwo \arrow \typ}}{\judgT{\tctx, \var:\typstwo}{\tmtwo}{ \typ}} & \infer{\judgTs{\tctx}{\tmsthree}{\typstwo}}{}}
\]
By Substitution Lemma \ref{lem:sub}, $\judgT{\tctx, \var:\typstwo}{\tmtwo}{ \typ}$ implies 
$\judgT{\tctx}{\tmtwo{\sub{\var^\typstwo}{\tmsthree}}}{ \typ}$. \\
Otherwise, it can be $\tm = \lam{\var^\typ}{\tmtwo}$, where $\tmtwo \toT \tmtwo'$, or 
$\tm=\app{\tmtwo}{\tmsthree}$, where either $\tmtwo \toT \tmtwo'$ or $\tmsthree \toT \tmsthree'$; in all these cases the proof follows by \ih.
\item $\tms \toT \tms'$ means $\tms=\set{\tm_1,..,\tm_n}$, and $\tms'=\set{\tm'_1,...,\tm'_n}$, and there is $i$ such that 
$\tm_i \toT \tm'_i$, and $\tm'_j=\tm_j$, for $j\not=i$. Then the proof follows by \ih. \qedhere
\end{enumerate}
\end{proof}

\subsubsection{Confluence}
\label{appendix:section:CalcLamT_confluence}

Confluence will be proved by the methodology of parallel reduction.

\begin{defn}
\begin{enumerate}

\item The parallel reduction $\pT$ is defined on terms and set-terms by the following rules, closed under contexts.\\
$
\begin{array}{lcll}
\var^\typ & \pT &\var^\typ &\\
\lam{\var^\typs}{\tm} &\pT & \lam{\var^\typs}{\tm'} &{\tt if} \quad\tm \pT \tm'\\
\tm \tmstwo & \pT & \tm' \tmstwo' &{\tt if} \quad\tm \pT \tm'\quad {\tt and} \quad\tmstwo \pT \tmstwo'\\
(\lam{\var^\typs}{\tm})\tmstwo &\pT& \tm'\sub{\var^\typs}{\tmstwo'}&{\tt if} \quad\tm \pT \tm' \quad {\tt and} \quad\tmtwo \pT \tmstwo'\\
\set{\tm_1, \hdots,\tm_n} & \pT & \set{\tm'_1, \hdots,\tm'_n} & {\tt if}\quad \tm_i \pT \tm'_i (1 \leq i \leq n)
\end{array}$ 

Note that $\pT$ is non-deterministic, since, when the term is a $\toT$ redex, either the third or the fourth rule can be applied. Roughly speaking, $\pT$ corresponds to reduce simultaneously a subset of the visible redexes by the $\toT$ reduction.
\item The complete development of a term $\tm$, is the following function, which denotes the result of the simultaneous reduction of all visible redexes in a term. 
 
$
  \begin{array}{rcl}
    \comp{\var^\typ}
    & \eqdef &
    \var^\typ
  \\
    \comp{\lam{\var^\typs}{\tm}}
    & \eqdef &
    \lam{\var^\typs}{\comp{\tm}}
  \\
 \comp{\app{(\lam{\var^\typs}{\tm})}{ \tmstwo}} & \eqdef & \app{\comp{\tm} \sub{\var^\typs}{\comp{\tmstwo}}}
  \\
    \comp{\app{\tm}{\tmstwo}}
    & \eqdef & 
      \app{\comp{\tm}}{\comp{\tmstwo}}\quad \text{if $\tm$ is not an abstraction}
  \\
    \comp{\set{\tm_1,\hdots,\tm_n}}
    & \eqdef &
    \set{\comp{\tm_1},\hdots,\comp{\tm_n}}
  \end{array}
$

\end{enumerate}
\end{defn}

 The $\pT$ reduction enjoys the following properties, whose proof is easy.
 \begin{prop}\label{propp:par}
 \begin{enumerate}
  \item $\tm \toT \tm'$ implies $\tm \pT \tm'$.
  \item $\tm \pT \tm'$ implies $\tm \toTs \tm'$.
  \item \label{three} $\toTs $ is the transitive closure of $\pT$.
 \end{enumerate}
 \end{prop}
 
\begin{lem}[Substitution]\label{lemmap:sub}
Let $\tm \pT \tm'$ and $\tmstwo \pT \tmstwo'$. Then $\tm \sub{\var^\typs}{\tmstwo} \pT \tm' \sub{\var^\typs}{\tmstwo'}$.
\end{lem}
\begin{proof} By induction on
 $\tm$. \end{proof}
 
 \begin{lem}\label{lem:develop}
 $\tm \pT \tm'$ implies $\tm' \pT \comp{\tm}$.
 \end{lem}
 \begin{proof}
 By induction on $\tm$. 
 Let $\tm = (\lambda \var^\typs.\tmthree)\tmstwo$. So $\tm'$ is either $(\lambda \var^\typs.\tmthree')\tmstwo'$ or $\tmthree' \sub{\var^\typs}{\tmstwo'}$, where $\tmstwo \pT \tmstwo'$ and $\tmthree \pT \tmthree'$. By \ih, $\tmthree' \pT \comp{\tmthree}$ and $\tmstwo' \pT \comp{\tmstwo}$. In both cases, $\tm' \toT \comp{\tmthree}\sub{\var^\typs}{ \comp{\tmstwo}}=\comp{\tm}$, in the former case by \ih, in the latter by \ih and Lemma \ref{lemmap:sub}.   The other cases come directly from \ih.
  \end{proof}
  
 The $\pT$ reduction enjoys the diamond property, i.e.,
 \begin{lem}\label{lemmap:diamond}
 If $\tm \pT \tm_1$ and $\tm \pT \tm_2$, then there is $\tm_3$ such that $\tm_1 \pT \tm_3$ and $\tm_2 \pT \tm_3$.
 \end{lem}
 \begin{proof}
By Lemma \ref{lem:develop}, $\tm_3 =\comp{\tm}$.
 \end{proof}

\confluenceprop*
\begin{proof}
  By Proposition \ref{propp:par}.\ref{three} and Lemma~\ref{lemmap:diamond}.
\end{proof}

\subsection{Proofs of Section~\ref{lambdaTm} --- Strong normalization via a decreasing measure}
\label{appendix:lambdaTm}
In this section we give detailed proofs of the results
about the $\CalcLamTm$-calculus stated in Section~\ref{lambdaTm}.

\subsubsection{Subject reduction}
\label{appendix:section:CalcLamTm_subject_reduction}

\begin{lem}[Substitution]
  \label{CalcLamTm_substitution}
  Let $\judgTm{\tctx,x:\typs}{\tm}{\typtwo}$
    and $\judgTms{\tctx}{\tmstwo}{\typs}$.
    Then $\judgTm{\tctx}{\tm\sub{\var^\typs}{\tmstwo}}{\typtwo}$.
  \begin{proof}
 By induction on $\tm$, generalizing the statement to set-terms. We consider only the wrapper case, since the remaining cases are exactly like those in Lemma~\ref{lem:sub}. Let $\tm = \tm_1\wrap{\tms_2}$. Then
      $\judgTm{\tctx,\var:\typs}{\tm_1\wrap{\tms_2}}{\typtwo}$
      concludes with $\ruleTmWrap$, so by inversion we have that
      $\judgTm{\tctx,\var:\typs}{\tm_1}{\typtwo}$ and
      $\judgTms{\tctx,\var:\typs}{\tms_2}{\typsthree}$.
      By \ih, $\judgTm{\tctx}{\tm_1\sub{\var^\typs}{\tmstwo}}{\typtwo}$
      and $\judgTms{\tctx}{\tms_2\sub{\var^\typs}{\tmstwo}}{\typsthree}$.
      By $\ruleTmWrap$ it follows that
      $\judgTm{\tctx}{\tm_1\sub{\var^\typs}{\tmstwo}\wrap{\tms_2\sub{\var^\typs}{\tmstwo}}} {\typtwo}$.
      By definition of substitution
      $(\tm_1\wrap{\tms_2})\sub{\var^\typs}{\tmstwo} =
      \tm_1\sub{\var^\typs}{\tmstwo}\wrap{\tms_2\sub{\var^\typs}{\tmstwo}}$,
      hence we conclude that
      $\judgTm{\tctx}{(\tm_1\wrap{\tms_2})\sub{\var^\typs}{\tmstwo}}{\typtwo}$.
  \end{proof}
\end{lem}

\calclamtmreduction*

\begin{proof}
Simultaneously by induction on the reduction relation $\toTm$.
\begin{enumerate}
\item In the case
  $
  \tm =
  (\lam{\var^\typstwo}{\tmtwo})\sctx\,\tmsthree
  \toTm
  \tmtwo\sub{\var^\typstwo}{\tmsthree}\wrap{\tmsthree}\sctx
  = \tm'
  $
  we have that $\tm$ must be typed with the following derivation scheme, where $n$
  is the amount of wrappers in $\sctx$:
  $$
  \infer[\ruleTmArre]{
    \judgTm{\tctx}{(\lam{\var^\typstwo}{\tmtwo})\sctx\,\tmsthree}{\typthree}}{
    \infer[\ruleTmWrap]{
      \judgTm{\tctx}{\lam{\var^\typstwo}{\tmtwo}\sctx}{\typstwo\arrow\typthree}}{
      \infer[\ruleTmWrap]{
        \vdots}{
          \infer[\ruleTmArri]{
            \judgTm{\tctx}{\lam{\var^\typstwo}{\tmtwo}}{\typstwo\arrow\typthree}}{
            \judgTm{\tctx,\var:\typstwo}{\tmtwo}{\typthree}}
          &
          \Pi_1}
        &
        \Pi_n}
    &
    \judgTm{\tctx}{\tmsthree}{\typstwo}}
  $$
  Hence from Lemma~\ref{CalcLamTm_substitution}, since
  $\judgTm{\tctx,\var:\typstwo}{\tmtwo}{\typthree}$ and
  $\judgTm{\tctx}{\tmsthree}{\typstwo}$, we obtain
  $\judgTm{\tctx}{\tmtwo\sub{\var^\typstwo}{\tmsthree}}{\typthree}$.
  Then from the corresponding $n+1$ applications of the $\ruleTmWrap$ we obtain
  $\judgTm{\tctx}{\tmtwo\sub{\var^\typstwo}{\tmsthree}\wrap{\tmsthree}\sctx}
  {\typthree}$.
\item If $t = \lam{\var^\typstwo}{\tmtwo} \toTm
  \lam{\var^\typstwo}{\tmtwo'} = \tm'$ with $\tmtwo \toTm \tmtwo'$,
  then $\typ = \typstwo\arrow\typthree$ and
  $\judgTm{\tctx}{\lam{\var^\typstwo}{\tmtwo}}{\typstwo\arrow\typthree}$.
  By inversion of $\ruleTmArri$
  it follows that $\judgTm{\tctx,\var:\typstwo}{\tmtwo}{\typthree}$.
  Then by IH we have that $\judgTm{\tctx,\var:\typstwo}{\tmtwo'}{\typthree}$,
  and by $\ruleTmArri$ we conclude
  $\judgTm{\tctx}{\lam{\var^\typstwo}{\tmtwo'}}{\typstwo\arrow\typthree}$.
\item If $\tm = \app{\tmtwo}{\tmsthree} \toTm \app{\tmtwo'}{\tmsthree} = \tm'$
  with $\tmtwo \toTm \tmtwo'$,
  then we have that $\judgTm{\tctx}{\app{\tmtwo}{\tmsthree}}{\typ}$ and,
  by inversion of $\ruleTmArre$, it follows that
  $\judgTm{\tctx}{\tmtwo}{\typstwo\arrow\typ}$ and
  $\judgTms{\tctx}{\tmsthree}{\typstwo}$.
  Then by IH we have that $\judgTm{\tctx}{\tmtwo'}{\typstwo\arrow\typ}$,
  and by $\ruleTmArre$ we conclude
  $\judgTm{\tctx}{\app{\tmtwo'}{\tmsthree}}{\typ}$.
\item If $\tm = \app{\tmtwo}{\tmsthree} \toTm \app{\tmtwo}{\tmsthree'} = \tm'$
  with $\tmsthree \toTm \tmsthree'$,
  then we have that $\judgTm{\tctx}{\app{\tmtwo}{\tmsthree}}{\typ}$ and,
  by inversion of $\ruleTmArre$, it follows that
  $\judgTm{\tctx}{\tmtwo}{\typstwo\arrow\typ}$ and
  $\judgTms{\tctx}{\tmsthree}{\typstwo}$.
  Then by IH we have that $\judgTms{\tctx}{\tmsthree'}{\typstwo}$,
  and by $\ruleTmArre$ we conclude
  $\judgTm{\tctx}{\app{\tmtwo}{\tmsthree'}}{\typ}$.
\item If $\tms = \set{\tm_1, \hdots, \tm_i, \hdots, \tm_n}
  \toTm \set{\tm_1, \hdots, \tm_i', \hdots, \tm_n} = \tms'$
  with $i \in 1..n$, $n \geq 1$ and $\tm_i \toTm \tm_i'$,
  then we have that
  $\judgTms{\tctx}{\set{\tm_1,\hdots,\tm_i,\hdots,\tm_n}}
  {\set{\typ_1,\hdots,\typ_i,\hdots\typ_n}}$.
  By inversion of rule $\ruleTmMulti$, it follows that
  $\judgTm{\tctx}{\tm_1}{\typ_1},\hdots,
  \judgTm{\tctx}{\tm_i}{\typ_i},\hdots,
  \judgTm{\tctx}{\tm_n}{\typ_n}$.
  Then by IH we have that $\judgTm{\tctx}{\tm_i'}{\typ_i}$,
  and by $\ruleTmMulti$ we conclude
  $\judgTms{\tctx}{\set{\tm_1,\hdots,\tm_i',\hdots,\tm_n}}
  {\set{\typ_1,\hdots,\typ_i,\hdots\typ_n}}$.
\item If $\tm = \bin{\tmtwo}{\tmsthree} \toTm \bin{\tmtwo'}{\tmsthree} = \tm'$
  with $\tmtwo \toTm \tmtwo'$,
  then we have that $\judgTm{\tctx}{\bin{\tmtwo}{\tmsthree}}{\typ}$ and,
  by inversion of $\ruleTmWrap$, it follows that
  $\judgTm{\tctx}{\tmtwo}{\typ}$ and
  $\judgTms{\tctx}{\tmsthree}{\typstwo}$.
  Then by IH we have that $\judgTm{\tctx}{\tmtwo'}{\typ}$,
  and by $\ruleTmWrap$ we conclude
  $\judgTm{\tctx}{\bin{\tmtwo'}{\tmsthree}}{\typ}$.
\item If $\tm = \bin{\tmtwo}{\tmsthree} \toTm \bin{\tmtwo}{\tmsthree'} = \tm'$
  with $\tmsthree \toTm \tmsthree'$,
  then we have that $\judgTm{\tctx}{\bin{\tmtwo}{\tmsthree}}{\typ}$ and,
  by inversion of $\ruleTmWrap$, it follows that
  $\judgTm{\tctx}{\tmtwo}{\typ}$ and
  $\judgTms{\tctx}{\tmsthree}{\typstwo}$.
  Then by IH we have that $\judgTms{\tctx}{\tmsthree'}{\typstwo}$,
  and by $\ruleTmWrap$ we conclude
  $\judgTm{\tctx}{\bin{\tmtwo}{\tmsthree'}}{\typ}$.
  \qedhere
\end{enumerate}
\end{proof}

\subsubsection{Confluence}
\label{appendix:section:CalcLamTm_confluence}

We extend the proof of Proposition~\ref{prop:confluence} by adding the cases
involving wrappers and modifying those involving $\LamTsymbol-$redexes by
$\LamTmsymbol-$redexes (\ie allowing a list of wrappers $\sctx$ between the
w-abstraction and the argument).
\begin{defn}\quad
\begin{enumerate}
\item
  The parallel reduction $\pTm$ is extended by modifying the $\beta$ case to
  allow lists of wrappers, and by adding the congruence cases of wrappers and
  lists of wrappers:
  \[
  \begin{array}{lcll}
    (\lam{\var^\typs}{\tm})\sctx\tmstwo
    & \pTm
    & \tm'\sub{\var^\typs}{\tmstwo'}\wrap{\tmstwo'}\sctx'
    & \text{if} \ \ \tm \pTm \tm'
      ,\ \tmtwo \pTm \tmstwo'\ \
      \text{and} \ \ \sctx \pTm \sctx' \\
    \tm\wrap{\tmstwo} & \pTm & \tm'\wrap{\tmstwo'}
    & \text{if} \ \ \tm \pTm \tm'
      \ \ \text{and} \ \ \tmstwo \pTm \tmstwo' \\
    \wlist & \pTm & \wlist[\tms']
    & \text{if} \ \ \tms_i\pTm \tms'_i \ \ \text{for all} \ \ i \in 1..n \\
  \end{array}
  \]
\item The complete development of a term $\tm$ is extended by modifying the
  $\beta$ case to allow lists of wrappers, and by adding the congruence cases of
  wrappers and lists of wrappers
  \[
  \begin{array}{rcl}
    \comp{(\lam{\var^\typs}{\tm})\sctx\tmstwo}
    & \eqdef
    & \comp{\tm} \sub{\var^\typs}{\comp{\tmstwo}} \wrap{\comp{\tmstwo}} \comp{\sctx}
  \\
    \comp{\tm\wrap{\tmstwo}}
    & \eqdef
    & \comp{\tm}\wrap{\comp{\tmstwo}}
  \\
    \comp{\wlist}
    & \eqdef &
    \ctxhole\wrap{\comp{\tm_1}},\hdots,\wrap{\comp{\tm_n}}
  \end{array}
  \]
\end{enumerate}
\end{defn}
As well as $\pT$ reduction, $\pTm$ reduction enjoys the following properties:
\begin{prop}[Properties of $\pTm$]
  \label{LamTm:lemma:par}\quad
  \begin{enumerate}
    \item $\tm \toTm \tm'$ implies $\tm \pTm \tm'$.
    \item $\tm \pTm \tm'$ implies $\tm \toTm^* \tm'$.
    \item \label{LamTm:lemma:par:transclosure}
      $\toTm^*$ is the transitive closure of $\pTm$.
  \end{enumerate}
\end{prop}

\begin{lem}[Substitution]\label{LamTm:lemma:subs}
  $\tm \sub{\vartwo^\typs}{\tmstwo} \sub{\var^\typstwo}{\tmsthree}
  = \tm \sub{\var^\typstwo}{\tmsthree}
  \sub{\vartwo^\typs}{\tmstwo \sub{\var^\typstwo}{\tmsthree}}$
  if $\vartwo \not \in \fv \tmsthree$
\end{lem}
\begin{proof}
  By induction on $\tm$.
\end{proof}

\begin{lem}[Compatibility of substitution and parallel reduction]
  \label{LamTm:lemma:parsub}
  Let $\tm \pTm \tm'$ and $\tmstwo \pTm \tmstwo'$.
  Then $\tm \sub{\var^\typs}{\tmstwo} \pTm \tm' \sub{\var^\typs}{\tmstwo'}$.
\end{lem}
\begin{proof}
By induction on $\tm$, where the interesting case is that of
$\tm = (\lam{\vartwo^\typs}{\tmthree})\sctx\tmsfour$. Then $\tm'$ depends on
whether the head $\symTm-$redex is contracted or not. In both cases it holds that
$\tmthree \pTm \tmthree'$, $\tmsfour \pTm \tmsfour'$, and $\sctx \pTm \sctx'$.
\begin{enumerate}
  \item If $\tm' = (\lam{\vartwo^\typstwo}{\tmthree'})\sctx'\tmsfour'$, then, by
    definition of substitution and \ih, we have that
    \[
      \begin{array}{rll}
        ((\lam{\vartwo^\typstwo}{\tmthree})
          \sctx
          \tmsfour)
          \sub{\var^\typs}{\tmstwo}
        & = & (\lam{\vartwo^\typstwo}{\tmthree \sub{\var^\typs}{\tmstwo}})
          \sctx\sub{\var^\typs}{\tmstwo}
          \tmsfour\sub{\var^\typs}{\tmstwo}\\
        & \pTm
        & (\lam{\vartwo^\typstwo}{\tmthree' \sub{\var^\typs}{\tmstwo'}})
          \sctx'\sub{\var^\typs}{\tmstwo'}
          \tmsfour'\sub{\var^\typs}{\tmstwo'}\\
        & = & ((\lam{\vartwo^\typstwo}{\tmthree'})
            \sctx'
            \tmsfour')
            \sub{\var^\typs}{\tmstwo'}
      \end{array}
    \]
  \item If $\tm' = \tmthree' \sub{\vartwo^\typstwo}{\tmsfour'} \wrap{\tmsfour'}
    \sctx'$, then, by definition of substitution, \ih, and
    Lemma~\ref{LamTm:lemma:subs}, we have that
    \[
      \begin{array}{rll}
        ((\lam{\vartwo^\typstwo}{\tmthree})
        \sctx
        \tmsfour)
        \sub{\var^\typs}{\tmstwo}
        & = & (\lam{\vartwo^\typstwo}{\tmthree\sub{\var^\typs}{\tmstwo}})
              \sctx\sub{\var^\typs}{\tmstwo}
              \tmsfour\sub{\var^\typs}{\tmstwo}\\
        & \pTm & \tmthree' \sub{\var^\typs}{\tmstwo'}
                 \sub{\vartwo^\typstwo}{\tmsfour' \sub{\var^\typs}{\tmstwo'}}
                 \wrap{\tmsfour' \sub{\var^\typs}{\tmstwo'}}
                 \sctx' \sub{\var^\typs}{\tmstwo'}\\
        & = & \tmthree' \sub{\vartwo^\typstwo}{\tmsfour'} \sub{\var^\typs}{\tmstwo'}
              \wrap{\tmsfour' \sub{\var^\typs}{\tmstwo'}}
              \sctx' \sub{\var^\typs}{\tmstwo'}\\
        & = & (\tmthree' \sub{\vartwo^\typstwo}{\tmsfour'} \wrap{\tmsfour'} \sctx')
              \sub{\var^\typs}{\tmstwo}
      \end{array}
    \]
  \end{enumerate}
\end{proof}

\begin{lem}[Composition lemma for wrappers]
  If $\tm \pTm \tm'$ and $\sctx \pTm \sctx'$, then $\tm \sctx \pTm \tm' \sctx'$.
\end{lem}
\begin{proof}
  Let $\sctx = \wlist$. We proceed by induction on $n$.
  \begin{enumerate}
  \item If $n=0$, then $\tm \sctx \tm \ctxhole = \tm \pTm \tm'
    = \tm' \ctxhole = \tm' \sctx'$.
  \item If $n=k+1$, then by \ih we have that
    $\tm \ctxhole \wrap{\tm_1} \hdots \wrap{\tm_k} \pTm
    \tm' \ctxhole \wrap{\tm'_1} \hdots \wrap{\tm'_k}$, hence
    \[
      \begin{array}{rll}
        \tm \sctx
        & = & \tm \wrap{\tm_1} \hdots \wrap{\tm_k} \wrap{\tm_{k+1}}\\
        & = & \tm \wrap{\tm_1} \hdots \wrap{\tm_k} \wrap{\tm_{k+1}}\\
        & \pTm & \tm' \wrap{\tm'_1} \hdots \wrap{\tm'_k} \wrap{\tm'_{k+1}}\\
        & = & \tm' \ctxhole \wrap{\tm'_1} \hdots \wrap{\tm'_k} \wrap{\tm'_{k+1}}\\
        & = & \tm' \sctx'
      \end{array}
    \]
  \end{enumerate}
\end{proof}

\begin{lem}\label{LamTm:lemma:dev}
$\tm \pTm \tm'$ implies $\tm' \pTm \comp{\tm}$.
\end{lem}
\begin{proof}
  By induction on $\tm$. The interesting case is that of $\tm = (\lambda
  \var^\typs.\tmthree)\sctx\tmstwo$. Then $\tm'$ depends on whether the head
  $\symTm$-redex is contracted or not. In both cases it holds that $\tmthree
  \pTm \tmthree'$, $\tmstwo \pTm \tmstwo'$, and $\sctx \pTm \sctx'$.
\begin{enumerate}
\item If $\tm' = (\lambda \var^\typs.\tmthree')\sctx'\tmstwo'$, then by \ih we
  have that
  \[
    \begin{array}{rll}
      (\lambda \var^\typs.\tmthree)\sctx\tmstwo
      & \pTm & (\lambda \var^\typs.\tmthree')
               \sctx'
               \tmstwo'\\
      & \pTm & (\lambda \var^\typs.\comp{\tmthree})
               \comp{\sctx}
               \comp{\tmstwo}\\
      & \toTm & \comp{\tmthree} \sub{\var^\typs}{\comp{\tmstwo}}
                \wrap{\comp{\tmstwo}}
                \comp{\sctx}\\
      & = & \comp{(\lambda \var^\typs.\tmthree)\sctx\tmstwo}
    \end{array}
  \]
\item If $\tm' = \tmthree' \sub{\var^\typs}{\tmstwo'} \wrap{\tmstwo'} \sctx'$,
  then by \ih and Lemma~\ref{LamTm:lemma:parsub} we have that
  \[
    \begin{array}{rll}
      (\lambda \var^\typs.\tmthree)\sctx\tmstwo
      & \pTm & \tmthree' \sub{\var^\typs}{\tmstwo'} \wrap{\tmstwo'} \sctx'\\
      & \pTm & \comp{\tmthree} \sub{\var^\typs}{\comp{\tmstwo}}
               \wrap{\comp{\tmstwo}} \comp{\sctx}\\
      & = & \comp{(\lambda \var^\typs.\tmthree)\sctx\tmstwo}
    \end{array}
  \]
\end{enumerate}
The remaining cases come directly from \ih.
\end{proof}

The $\pTm$ reduction enjoys the diamond property, i.e.,
\begin{lem}\label{LamTm:lemma:diamond}
If $\tm \pTm \tm_1$ and $\tm \pTm \tm_2$, then there is $\tm_3$ such that $\tm_1 \pTm \tm_3$ and $\tm_2 \pTm \tm_3$.
\end{lem}
\begin{proof}
By Lemma~\ref{LamTm:lemma:dev}, $\tm_3 =\comp{\tm}$.
\end{proof}

\lamtmconfluence*
\begin{proof}
  By Lemmas~\ref{LamTm:lemma:par}.\ref{LamTm:lemma:par:transclosure} and~\ref{LamTm:lemma:diamond}.
\end{proof}

\end{document}